\documentclass[11pt]{article}

\usepackage{amsmath, amssymb, amsthm}
\usepackage{fullpage}
\usepackage{graphics}
\usepackage{clrscode4e}
\usepackage{environ}
\usepackage{framed}

\usepackage[labelfont=bf]{caption}

\newcommand\remove[1]{}

\theoremstyle{plain}
\newtheorem{lemma}{Lemma}[section]
\newtheorem{theorem}[lemma]{Theorem}
\newtheorem{corollary}[lemma]{Corollary}

\newtheorem{rem}{Remark}[section]
\newtheorem{thm}[lemma]{Theorem}

\theoremstyle{definition}
\newtheorem{definition}[lemma]{Definition}

\DeclareMathOperator{\poly}{poly}
\newcommand\card[1]{\left| #1 \right|}

\newcommand\sett[2]{\left\{ \left. #1 \;\right\vert #2 \right\}}

\newcommand\set[1]{{\left\{ #1 \right\}}}

\renewcommand\emptyset{\phi}

\newcommand\Prob[2]{{\Pr_{#1}\left[ {#2} \right]}}

\newcommand\Expc[2]{{\mathop{\bf E}_{#1}\left[ {#2} \right]}}

\newcommand\ceil[1]{\lceil{#1}\rceil}

\newcommand\defeq{\doteq}
\newcommand\floor[1]{\lfloor{#1}\rfloor}

\newcommand\N{\mathbb{N}}
\newcommand\field{\mathbb{F}}
\newcommand\dims{m}
\newcommand\vecspace{\field^\dims}
\newcommand\degree{d}

\newif\ifrandom
\randomtrue

\NewEnviron{algbox}[3]{
	\begin{figure}[th]
	\begin{center}
	\begin{minipage}{#1}
	\begin{framed}
	\vspace{-1.25\baselineskip}
	\begin{codebox}
		\BODY
	\end{codebox}
	\vspace{-1.25\baselineskip}
	\end{framed}
	\end{minipage}
	\end{center}
	\caption{#3}
	\label{alg:#2}
	\end{figure}
}

\author{
	Ofer Grossman\thanks{
		{\tt ogrossma@mit.edu}.
		Department of Mathematics,
		MIT.
	} \and
	Dana Moshkovitz\thanks{
		{\tt dmoshkov@csail.mit.edu}.
		Department of Electrical Engineering and Computer Science,
		MIT.
		This material is based upon work supported by the National Science Foundation under grants number 1218547 and 1452302. 
}
}

\begin{document}
\pagenumbering{gobble}

\title{Amplification and Derandomization Without Slowdown}

\maketitle

\begin{abstract}
We present techniques for decreasing the error probability of randomized algorithms and for converting randomized algorithms to deterministic (non-uniform) algorithms. 
Unlike most existing techniques that involve repetition of the randomized algorithm and hence a slowdown, our techniques produce algorithms with a similar run-time to the original randomized algorithms. 
The amplification technique is related to a certain stochastic multi-armed bandit problem.
The derandomization technique -- which is the main contribution of this work -- points to an intriguing connection between derandomization and sketching/sparsification.

We demonstrate the techniques by showing the following applications:
\begin{enumerate}
\item \textbf{Dense Max-Cut:} A Las Vegas algorithm that given a $\gamma$-dense $G=(V,E)$ that has a cut containing $1-\varepsilon$ fraction of the edges, finds a cut that contains $1-O(\varepsilon)$ fraction of the edges. The algorithm runs in time $\tilde{O}(\card{V}^2(1/\varepsilon)^{O(1/\gamma^2 + 1/\varepsilon^2)})$ and has error probability exponentially small in $\card{V}^2$. It also implies a deterministic non-uniform algorithm with the same run-time (note that the input size is $\Theta(\card{V}^2)$). 
 
\item \textbf{Approximate Clique:} A Las Vegas algorithm that given a graph $G=(V,E)$ that contains a clique on $\rho\card{V}$ vertices, and given $\varepsilon>0$, finds a set on $\rho\card{V}$ vertices of density at least $1-\varepsilon$. The algorithm runs in time $\tilde{O}(\card{V}^2 2^{O(1/(\rho^3\varepsilon^2))})$ and has error probability exponentially small in $\card{V}$. We also show a deterministic non-uniform algorithm with the same run-time.

\item \textbf{Free Games:} A Las Vegas algorithm and a non-uniform deterministic algorithm that given a free game (constraint satisfaction problem on a dense bipartite graph) with value at least $1-\varepsilon_0$ and given $\varepsilon>0$, find a labeling of value at least $1-\varepsilon_0-\varepsilon$. The error probability of the randomized algorithm is exponentially small in the number of vertices and labels. The run-time of the algorithms is similar to that of algorithms with constant error probability.

%The algorithm runs in time $\tilde{O}(\card{X} %\card{\Sigma}^{(1/\varepsilon^2)\log(\card{\Sigma}/\varepsilon^2))+1})$ and has %error probability exponentially small in the number of vertices. We also show a %deterministic non-uniform algorithm that runs in time %$\tilde{O}(\card{X}\card{Y}\card{\Sigma}+\card{X} %\card{\Sigma}^{(1/\varepsilon^2)\log(\card{\Sigma}/\varepsilon^2))+1})$. 

\item \textbf{From List Decoding To Unique Decoding For Reed-Muller Code:} A randomized algorithm with error probability exponentially small in the input size that given a word $f$ and $0<\epsilon ,\rho<1$ finds a short list such that every low degree polynomial that is $\rho$-close to $f$ is $(1-\epsilon)$-close to one of the words in the list. The algorithm runs in nearly linear time in the input size, and implies a deterministic non-uniform algorithm with similar run-time. 
The run-time of our algorithms compares with that of the most efficient algebraic algorithms, but our algorithms are combinatorial and much simpler. 
\end{enumerate}
\end{abstract}

\newpage
\pagenumbering{arabic}
\setcounter{page}{1}

\section{Introduction}

\subsection{Amplification}\label{ss:coin}

Given a randomized algorithm that runs in time $t$ and has error probability $1/3$, can we find a randomized algorithm that runs in similar time and has a substantially smaller error probability $e^{-\Omega(k)}$? One can achieve such a low error probability by repeating the algorithm $k$ times. 
%Impagliazzo and Zuckerman showed that this can even be implemented in a %randomness-efficient manner~\cite{IZ}. 
However, the resulting algorithm is slower by a factor of $k$ than the original algorithm, which is a significant slowdown when $k$ is large (for instance, consider $k$ that equals the input size $n$, or equals $n^{\epsilon}$ for some constant $\epsilon>0$). 
In this work we show that in many situations one can decrease the error probability of the algorithm to $e^{-\Omega(k)}$ without any substantial slowdown. These situations occur when there is an additional randomized algorithm for evaluating the quality of the randomized choices of the algorithm that is more efficient than the overall algorithm. 
%This is the case for many problems, e.g., graph problems like \textsc{Max-Cut} %or \textsc{Clique} on dense graphs, \textsc{Minimum Spanning Tree} with its %linear-time deterministic verification algorithm~\cite{DRT} and %\textsc{Matrix-Multiplication} with its quadratic time verification algorithm.

We show how to capitalize on the existence of such a testing algorithm using an algorithm for a stochastic multi-armed bandit problem that we define.
In this problem, which we call the {\em biased coin problem}, there is a large pile of coins, and $2/3$ fraction of the coins are biased, meaning that they fall on heads with high probability. The coins are unmarked and the only way to discover information about a coin is to toss it. The task is to find one biased coin with certainty $1-e^{-\Omega(k)}$ using as few coin tosses as possible. The analogy between the biased coin problem and amplification is that the coins represent possible random choices of the algorithm, many of which are good. The task is to find one choice that is good with very high probability. Tossing a coin corresponds to testing the random choice of the algorithm.

\subsection{Derandomization}
%In the last decades randomized algorithms became pervasive. The idea is to give %up on certainty, either in the correctness of the output or in the run-time of %the algorithm, in order to obtain more efficient algorithms. 
What speed-up does randomization buy?
Impagliazzo and Wigderson~\cite{IW} showed that, under plausible hardness assumptions, randomness can only speed up a polynomial-time computation by a polynomial factor. Their deterministic algorithm, which invokes the randomized algorithm on randomness strings generated by enumerating over all possible seeds of a pseudorandom generator, slows down the run-time by at least a linear factor. To avoid the reliance on unproven assumptions, researchers typically use properties of the concrete randomized algorithms they wish to derandomize and design (or use off-the-shelf) pseudorandom generators for them (e.g., pairwise independent, $\varepsilon$-biased sets, $k$-wise-independent and almost $k$-wise independent; see, e.g.,~\cite{LW,NN,AGHP}). Here too derandomization slows down the run-time by at least a linear factor. 

A different derandomization method by Adleman~\cite{Adleman} uses amplification for derandomization and does not rely on any unproven assumptions. It generates a non-uniform deterministic algorithm by first decreasing the error probability of the randomized algorithm below $2^{-n}$ where $n$ is the input size. Then, there must exist a randomness string that works for all $2^n$ inputs, and this randomness string can be hard-wired to a non-uniform algorithm. Due to the slowdown in amplification discussed in the previous sub-section, this technique too slows down the run-time by a linear factor in $n$.  

There is a general derandomization method that typically does not increase the run-time significantly, namely {\em the method of conditional probabilities}. It is used, for instance, for finding an assignment that satisfies $7/8$ fraction of the clauses in a \textsc{3Sat} formula. However, this method works only in very special cases. % method of pessimistic estimators? 
In this work we'll be interested in derandomizing algorithms without slowing down the run-time significantly, in cases where the method of conditional expectations does not apply. 

Our derandomization method builds on Adleman's technique but avoids its slowdown, by using a new connection to sketching and sparsification. Briefly, the connection is as follows: consider a verifier that given an input and a randomness string for the randomized algorithm tests whether the outcome of the randomized algorithm is correct. If the verifier can perform its test with only a size-$n'$ sketch of the input (we call such a verifier an {\em oblivious verifier}), then Adleman's union bound can be performed over $2^{n'}$ representative inputs, rather than over $2^n$ inputs. This means that it suffices to amplify the error probability below $2^{-n'}$. This saving, together with the amplification technique discussed in the previous sub-section, allows us to derandomize without slowdown. 

\subsection{Context}

The main existing approach to derandomization -- the one based on pseudorandom generators -- focuses on shrinking the number of random strings. This is possible since the algorithm is limited (by its run-time or by the simple way it uses the randomness) and cannot distinguish the set of all randomness strings from a small subset of it (pseudorandom strings). In contrast, our approach focuses on shrinking the number of inputs one needs to argue about. We show that it's enough that a randomness string leads to a correct output for all sketches.

Crucially, we do not argue that the algorithm doesn't make use of its entire input, or that inputs with the same sketch are indistinguishable, or that inputs with the same sketch are not distinguished by the algorithm. The algorithms we consider depend on all of their input. 
Our argument relies on the existence of a {\em verifier} aimed at certifying that randomness is good for an input, and which doesn't distinguish between inputs with the same sketch. Surprisingly, we are able to devise sketches and design such oblivious verifiers for many natural algorithms.
%We believe that it might be the first time anyone takes advantage of the fact %(that could potentially be taken advantage of also in the context of %pseudorandom generators) that the distinguisher is limited to answering the %question: ``does this algorithm succeed on this input and randomness?'' %(existing work on pseudorandom generators takes advantage of this fact only in %a weak sense, by using that the distinguisher is limited computationally since %the algorithm is).

The sketch that the oblivious verifier uses can be hard to compute, and it may reveal to the verifier a correct output. The verifier need not (and generally will not be) efficient. The only requirement is that the number of bits in the sketch is small and that the verifier is deterministic (though the construction of the sketch can be probabilistic -- we only need existence of a sketch). 
Our applications include problems on dense graphs where sketching can be done using uniform samples.
We hope that the large body of work on sparsification and sketching (see, e.g., ~\cite{Karger,BK,TZ} and the many works that followed them) could be used for more sophisticated applications of our methods.

\remove{
It is interesting to note the duality between our derandomization and derandomization through pseudorandom generators. Both methods rely on the simplicity of the randomized algorithm and the matrix that describes its behavior. When considering pseudorandom generators we argue that a small collection of randomness strings (matrix columns) suffices, while when considering oblivious verifiers we argue that a small collection of inputs (matrix rows) suffice.  
}

\subsection{Non-Uniform Algorithms, Preprocessing and Amortization}

Our derandomization produces non-uniform algorithms, i.e., algorithms that are designed with a specific input size in mind. The knowledge of the input size is manifested by an ``advice'' string that depends on the input size. The size of the advice counts toward the run-time of the algorithm (so, for instance, advice that consists of the output for each possible input leads to an exponential-time algorithm). Equivalently, non-uniform algorithms are described as sequences of circuits, one for each input size. Sorting networks are an instance of non-uniform algorithms.
%Non-uniform algorithms are a standard model for efficient computation, %commonplace in complexity, cryptography and algorithms. 
%and, are, arguably, more realistic than the uniform model, in which the %algorithm cannot be designed with knowledge of the input size.     

In some cases non-uniform algorithms imply uniform algorithms with the same asymptotic run time. This is the case with \textsc{Matrix-Multiplication} and \textsc{Minimum-Spanning-Tree}~\cite{PeRa}. More generally, whenever a problem on inputs of size $n$ can be reduced to the same problem on $n/s$ inputs of size $s$ each for, say, $s=\log\log\log n$, a non-uniform algorithm for the problem implies a uniform algorithm. The uniform algorithm uses exhaustive search to find the advice for inputs of size $s$ (checking all possible advices, and for each, all possible inputs). It then uses the reduction to find the sub-problems and the non-uniform algorithm to solve the sub-problems.

Even when a reduction of this sort does not exist,
one can either designate the search for a good advice as a preprocessing phase after which the algorithm is correct on all inputs, or amortize the cost of searching for a good advice across inputs. If the space of possible advice strings contains $2^a$ possibilities (where $a$ can be as small as $O(\log n)$ if the space is the set of possible outputs of a pseudorandom generator), and one can amortize the cost over $2^a$ inputs, then one obtains the desired run-time uniformly, amortized.

%Even when deterministic non-uniform algorithms do not immediately imply %deterministic uniform algorithms, they can be thought of as stepping stones %toward deterministic uniform algorithm. 

%Alternatively, one can think of our derandomization method as a method not for %producing deterministic non-uniform algorithms, but as a method for producing %stronger randomized algorithms: ones where most randomness strings are good for %all inputs (as opposed to each input having a possibly different set of good %randomness strings). If one wants to apply them on many inputs, and the search %for a good randomness string can be amortized across different inputs.  %dominated by the run-time of   

\subsection{Applications}

We demonstrate our techniques with applications for \textsc{Max-Cut} on dense graphs, (approximate) \textsc{Clique} on graphs that contain large cliques, free games (constraint satisfaction problems on dense bipartite graphs), and reducing the Reed-Muller list decoding problem to its unique decoding problem. 
All our algorithms run in nearly linear time in their input size, and all of them beat the current state of the art algorithms in one aspect or another. 
The biggest improvement is in the algorithm for free games that is more efficient by orders of magnitude than the best deterministic algorithms.
The algorithm for \textsc{Max-Cut} can efficiently handle sparser graphs than the best deterministic algorithm, the algorithm for (approximate) \textsc{Clique} can efficiently handle smaller cliques than the best deterministic algorithm; and the algorithm for the Reed-Muller code achieves similar run-time as sophisticated algebraic algorithms despite being much simpler. In general, our focus is on demonstrating the utility and versatility of the techniques and not on obtaining the most efficient algorithm for each problem. In the open problems section we point to several aspects where we leave room for improvement.

\subsubsection{Max Cut on Dense Graphs} 

Given a graph $G=(V,E)$, a cut in the graph is defined by $C\subseteq V$. The value of the cut is the fraction of edges $e=(u,v)\in E$ such that $u\in C$ and $v\in V-C$. We say that a graph is $\gamma$-dense if it contains $\gamma\card{V}^2/2$ edges. For simplicity we assume that the graph is regular, so every vertex has degree $\gamma\card{V}$. 
Given a regular $\gamma$-dense graph that has a cut of value at least $1-\varepsilon$ for $\varepsilon < 1/4$, we'd like to find a cut of value roughly $1-\varepsilon$. Understanding this problem on general (non-dense) graphs is an important open problem: (a weak version of) the Unique Games Conjecture~\cite{Khot}. However, for dense graphs, it is possible to construct a cut of value $1-\varepsilon-\zeta$ efficiently~\cite{Vega,AKK95,GGR,MS}.
The best randomized algorithms are an algorithm of Mathieu and Schudy~\cite{MS} that runs in time $O(\card{V}^2 + 2^{O(1/\gamma^2\zeta^2)})$ and an algorithm of Goldreich, Goldwasser and Ron~\cite{GGR} that runs in time $O(\card{V}(1/\gamma\zeta)^{O(1/\gamma^2\zeta^2)} + (1/\gamma\zeta)^{O(1/\gamma^3\zeta^{3})})$ (Note that the algorithm of~\cite{GGR} runs in sub-linear time. This is possible since it is an Atlantic City algorithm).
Both algorithms have constant error probability. We obtain a Las Vegas algorithm with exponentially small error probability, and deduce a deterministic non-uniform algorithm. This is the simplest application of our techniques. It uses the biased coin algorithm, but does not require any sketches.
 
\begin{theorem}\label{t:max-cut} There is a Las Vegas algorithm that given a $\gamma$-dense graph $G$ that has a cut of value at least $1-\varepsilon$ for $\varepsilon < 1/4$, and given $\zeta < 1/4-\varepsilon$, finds a cut of value at least $1-\varepsilon - O(\zeta)$, except with probability exponentially small in $\card{V}^2$. The algorithm runs in time $\tilde{O}(\card{V}^2(1/\zeta)^{O(1/\gamma^2 + 1/\zeta^2)})$. It also implies a non-uniform deterministic algorithm with the same run-time.
\end{theorem}
Note that run-time $\Omega(\gamma\card{V}^2)$ is necessary for a deterministic algorithm, since the input size is $\gamma\card{V}^2$.
A deterministic $O(\card{V}^2\poly(1/\gamma\zeta) + 2^{\poly(1/\gamma\zeta)})$-time algorithm follows from a recent deterministic version of the Frieze-Kannan regularity lemma~\cite{Szeme,FrKa,DKMRS,DKMRS0,ADLRY}, however the $\poly(\cdot)$ term in the exponent hides large constant exponents. Therefore, our algorithm handles efficiently graphs that are sparser than those handled efficiently by the existing deterministic algorithm.

%% existing algorithms don't have sqrt{epsilon} - need to check!!!

\subsubsection{Approximate Clique} 

The input is $0<\varepsilon,\rho<1$ and an undirected graph $G=(V,E)$ for which there exists a set $C\subseteq V$, $\card{C}\geq \rho\card{V}$, that spans a clique. The goal is to find a set $D\subseteq V$, $\card{D}\geq \rho\card{V}$, whose edge density is at least $1-\varepsilon$, i.e., if $E(D)\subseteq E$ is the set of edges whose endpoints are in $D$, then $\card{E(D)}/\binom{\card{D}}{2}\geq 1-\varepsilon$.
Goldreich, Goldwasser and Ron~\cite{GGR} gave a randomized $O(\card{V} (1/\varepsilon)^{O(1/(\rho^3\varepsilon^2))})$ time algorithm for this problem with constant error probability (Note that this is a sub-linear time algorithm. This is possible since it is an Atlantic City algorithm). A deterministic $O(\card{V}^2\poly(1/\rho,1/\varepsilon) + 2^{\poly(1/\rho,1/\varepsilon)})$ time algorithm with worse dependence on $\rho$ and $\varepsilon$ follows from a deterministic version of the Frieze-Kannan regularity lemma~\cite{Szeme,FrKa,DKMRS,DKMRS0,ADLRY}. We obtain a randomized algorithm with exponentially small error probability in $\card{V}$, and use sketching to obtain a non-uniform deterministic algorithm. Our algorithms have better dependence in $\rho$ and $\varepsilon$ than the existing deterministic algorithm, and can therefore handle efficiently graphs with smaller cliques than the existing deterministic algorithm and output denser sets.

\begin{theorem}\label{t:clique} The following hold: 
\begin{enumerate}
\item There is a Las Vegas algorithm that given $0<\rho,\varepsilon<1$, and a graph $G=(V,E)$ with a clique on $\rho\card{V}$ vertices, finds a set of $\rho\card{V}$ vertices and density at least $1-\varepsilon$, except with probability exponentially small in $\card{V}$. The algorithm runs in time $\tilde{O}(\card{V}^2 2^{O(1/(\rho^3\varepsilon^2))})$. 

\item There is a deterministic non-uniform algorithm that given $0<\rho,\varepsilon<1$, and a graph $G=(V,E)$ with a clique on $\rho\card{V}$ vertices, finds a set of $\rho\card{V}$ vertices and density at least $1-\varepsilon$. The algorithm runs in time $\tilde{O}(\card{V}^2 2^{O(1/(\rho^3\varepsilon^2))})$.

\end{enumerate} 
\end{theorem}
The sketch for approximate clique consists of all the edges that touch a small random set of vertices. We show that such a sketch suffices to estimate the density of the sets considered by the algorithm and the quality of the random samples of the algorithm.

\subsubsection{Free Games} 

A {\em free game} $\mathcal{G}$ is defined by a complete bipartite graph $G=(X,Y,X\times Y)$, a finite alphabet $\Sigma$ and constraints $\pi_e \subseteq \Sigma\times\Sigma$ for all $e\in X\times Y$. For simplicity we assume $\card{X}=\card{Y}$. A labeling to the vertices is given by $f_X:X\to\Sigma$, $f_Y:Y\to\Sigma$. 
The value achieved by $f_X,f_Y$, denoted $val_{f_X,f_Y}(\mathcal{G})$, is the fraction of edges that are satisfied by $f_X$, $f_Y$, where an edge $e=(x,y)\in X\times Y$ is satisfied by $f_X$, $f_Y$ if $(f_X(x),f_Y(y))\in \pi_e$.
The value of the instance, denoted $val(\mathcal{G})$, is the maximum over all labelings $f_X:X\to\Sigma$, $f_Y:Y\to\Sigma$, of $val_{f_X,f_Y}(\mathcal{G})$.
Given a game $\mathcal{G}$ with value $val(\mathcal{G})\geq 1-\varepsilon$, the task is to find a labeling to the vertices $g_X:X\to\Sigma$, $g_Y:Y\to\Sigma$, that satisfies at least $1-O(\varepsilon)$ fraction of the edges.

Free games have been researched in the context of one round two prover games (see~\cite{Feige95errorreduction} and many subsequent works on parallel repetition of free games) and two prover AM~\cite{AIM}.
They unify a large family of problems on dense bipartite graphs obtained by considering different constraints. For instance, for \textsc{Max-2Sat} we have $\Sigma = \set{T,F}$, and $\pi_e$ contains all $(a,b)$ that satisfy $\alpha\vee\beta$ where $\alpha$ is either $a$ or $\neg a$ and $\beta$ is either $b$ or $\neg b$.
%\begin{itemize}
%\item Max-Cut: $\Sigma = \set{0,1}$. $\pi_e$ either contains $(a,b)$ such that %$a\neq b$ or such that $a=b$.
%\item $k$-Coloring: $\Sigma$ is a collection of $k$ colors. $\pi_e$ either %contains $(a,b)$ such that $a\neq b$ or such that $a=b$.
%\item 
%\end{itemize}
Note that on a small fraction of the edges the constraints can be ``always satisfied'', so one can optimize over any dense graph, not just over the complete graph (the density of the graph is crucial: if fewer than $\varepsilon \card{X}\card{Y}$ of the edges have non-trivial constraints, then any labeling satisfies $1-\varepsilon$ fraction of the edges). 

There are randomized algorithms for free games that have constant error probability~\cite{AKK95,AVKK,BMHS11,AIM}, as well as a derandomization that incurs a polynomial slowdown~\cite{AKK95}. In addition, deterministic algorithms for free games of value $1$ are known~\cite{MM-dense}.
We show a randomized algorithm with exponentially small error probability in $\card{X}\card{\Sigma}$ and a non-uniform deterministic algorithm whose running time is similar to that of the randomized algorithms with constant error probability.

\begin{theorem}\label{t:free} The following hold: 
\begin{enumerate}
\item There is a Las Vegas algorithm that given a free game $\mathcal{G}$ with vertex sets $X,Y$, alphabet $\Sigma$, and $val(\mathcal{G})\geq 1-\varepsilon_0$, and given $\varepsilon >0$, finds a labeling to the vertices that satisfies $1-\varepsilon_0-O(\varepsilon)$ fraction of the edges, except with probability exponentially small in $\card{X}\card{\Sigma}$. The algorithm runs in time $\tilde{O}(\card{X}\card{Y} \card{\Sigma}^{O((1/\varepsilon^2)\log(\card{\Sigma}/\varepsilon))})$. 

\item There is a deterministic non-uniform algorithm that given a free game $\mathcal{G}$ with vertex sets $X,Y$, alphabet $\Sigma$, and $val(\mathcal{G})\geq 1-\varepsilon_0$, and given $\varepsilon > 0$, finds a labeling to the vertices that satisfies $1-\varepsilon_0 - O(\varepsilon)$ fraction of the edges. The algorithm runs in time $\tilde{O}(\card{X}\card{Y} \card{\Sigma}^{O((1/\varepsilon^2)\log(\card{\Sigma}/\varepsilon))})$.

\end{enumerate} 
\end{theorem}

The sketch of a free game consists of the restriction of the game to a small random subset of $Y$. We show that the sketch suffices to estimate the value of the labelings considered by the algorithm and the random samples the algorithm makes.

\subsubsection{From List Decoding to Unique Decoding of Reed-Muller Code}

\begin{definition}[Reed-Muller code]
The Reed-Muller code defined by a finite field $\field$ and natural numbers $\dims$ and $\degree$ consists of all $\dims$-variate polynomials of degree at most $\degree$ over $\field$.
\end{definition}

Let $0<\epsilon<\rho<1$. In the list decoding to unique decoding problem for the Reed-Muller code, the input is a function $f:\vecspace\to\field$ and the goal is to output a list of $l=O(1/\epsilon)$ functions $g_1,\ldots,g_l:\vecspace\to\field$, such that for every $\dims$-variate polynomial $p$ of degree at most $\degree$ over $\field$ that agrees with $f$ on at least $\rho$ fraction of the points $x\in\vecspace$, there exists $g_i$ that agrees with $p$ on at least $1-\epsilon$ fraction of the points $x\in\vecspace$. 

There are randomized, self-correction-based, algorithms for this problem (see~\cite{STV} and the references there). There are also deterministic list decoding algorithms for the Reed-Solomon and Reed-Muller codes that can solve this problem: The algorithms of Sudan~\cite{Sudan} and Guruswami-Sudan~\cite{GuSu} run in large polynomial time, as they involve solving a system of linear equations and factorization of polynomials. There are efficient implementations of these algorithms that run in time $\tilde{O}(\card{\vecspace})$ (see~\cite{Alek} and the references there), but they involve deeper algebraic technique. Our contribution is simple, combinatorial, algorithms, randomized and deterministic, with nearly-linear run-time. This application too relies on the biased coin algorithm but does not require sketching.

\begin{theorem}\label{t:RM}
Let $\field$ be a finite field, let $\degree$ and $\dims>3$ be natural numbers and let $0<\rho,\epsilon<1$, such that $\degree\leq \card{\field}/10$, $\epsilon >\sqrt[3]{2/\card{\field}}$ and $\rho>\epsilon + 2\sqrt{\degree/\card{\field}}$. 
%\begin{enumerate}
%\item 
There is a randomized algorithm that given $f:\vecspace\to\field$, finds a list of $l=O(1/\rho)$ functions $g_1,\ldots,g_l:\vecspace\to\field$, such that for every $\dims$-variate polynomial $p$ of degree at most $\degree$ over $\field$ that agrees with $f$ on at least $\rho$ fraction of the points $x\in\vecspace$, there exists $g_i$ that agrees with $p$ on at least $1-\epsilon$ fraction of the points $x\in\vecspace$. The algorithm has error probability exponentially small in $\card{\vecspace}\log\card{\field}$ and it runs in time $\tilde{O}(\card{\vecspace}\poly(\card{\field}))$. It implies a deterministic non-uniform algorithm with the same run-time. 

%\item There is a deterministic non-uniform algorithm that given %$f:\vecspace\to\field$, finds a list of $l=\poly(\card{\field},\dims,1/\rho)$ %functions $g_1,\ldots,g_l:\vecspace\to\field$, such that for every %$\dims$-variate polynomial $p$ of degree at most $\degree$ over $\field$ that %agrees with $f$ on at least $\rho$ fraction of the points $x\in\vecspace$, %there exists $g_i$ that agrees with $p$ on at least $1-\epsilon$ fraction of %the points $x\in\vecspace$. The algorithm runs in time %$\tilde{O}(\card{\vecspace}\poly(\card{\field}))$. 

%\end{enumerate} 

\end{theorem}

Note that the standard choice of parameters for the Reed-Muller code has $\card{\field} = \poly\log\card{\vecspace}$, and in this case our algorithms run in nearly linear time $\tilde{O}(\card{\vecspace})$. 

%Moreover, if this algorithm is followed by some unique decoding algorithm, the %longer list in the deterministic case only slows down that algorithm by %poly-logarithmic factors.

% couldn't find these online:

% G. L. Feng, "Two fast algorithms in the Sudan decoding procedure,"
%in Proc. 37th Annu. Allerton Conf. Communication, Control and Computing,
%Monticello, IL, Oct. 1999, pp. 545–554.

% G. L. Feng and X. Giraud, "Fast algorithms in Sudan decoding procedure
%for Reed–Solomon codes," IEEE Trans. Inf. Theory, to be published.

% doesn't even require sketch!

\subsection{Previous Work} %% some general, some problem related

The biased coin problem introduced in Sub-section~\ref{ss:coin} is related to the stochastic multi-armed bandit problem studied in~\cite{EMM,MT}, however, in the latter there might be only one biased coin, whereas in our problem we are guaranteed that a constant fraction of the coins are biased. This makes a big difference in the algorithms one would consider for each problem and in their performance. In the setup considered by~\cite{EMM,MT} one has to toss all coins, and the algorithms focus on which coins to eliminate. In our setup it is likely that we find a biased coin quickly, and the focus is on certifying bias. In~\cite{EMM,MT} an $\Omega(k^2)$ lower bound is proved for the number of coin tosses needed to find a biased coin with probability $1-e^{-\Omega(k)}$, whereas we present an $\tilde{O}(k)$ upper bound for the case of a constant fraction of biased coins.   

The connection that we make between derandomization and sketching adds to a long list of connections that have been identified over the years between derandomization, compression, learning and circuit lower bounds, e.g., circuit lower bounds can be used for pseudorandom generators and derandomization~\cite{IW};
learning goes hand in hand with compression, and can be used to prove circuit lower bounds~\cite{FoKla};
simplification under random restrictions can be used to prove circuit lower bounds~\cite{Subbo} and construct pseudorandom generators~\cite{IMZ}. Sparsification of the distinguisher of a pseudorandom generator (e.g., for simple distinguishers like DNFs) can lead to more efficient pseudorandom generators and derandomizations~\cite{GMR}. 
Our connection differs from all those connections.
In particular, previous connections are based on pseudorandom generators, while our approach is dual and focuses on shrinking the number of inputs.
 
The idea of saving in a union bound by only considering representatives is an old idea with countless appearances in math and theoretical computer science, including derandomization (one example comes from the notion of an $\varepsilon$-net and its many uses; another example is~\cite{GMR} we mentioned above). Our contribution is in the formulation of an oblivious verifier and in designing sketches and oblivious verifiers.

Our applications have Atlantic City algorithms that run in sub-linear time and have a constant error probability.
There are works that aim to derandomize sub-linear time algorithms. Most notably, as mentioned before, there is a deterministic version of the Frieze-Kannan regularity lemma~\cite{Szeme,FrKa,DKMRS,DKMRS0,ADLRY}, which is relevant to some of our applications but not to others. Another work is~\cite{Zimand} that generates deterministic {\em average case} algorithms for decision problems with certain sub-linear run time while incurring a slowdown.

\section{Preliminaries}

\subsection{Conventions and Inequalities}

\begin{lemma}[Chernoff bounds]
Let $X_1,\ldots,X_n$ be i.i.d random variables taking values in $\set{0,1}$. Let $\varepsilon >0$. Then,
$$\Prob{}{\frac{1}{n}\sum X_i \geq \frac{1}{n}\sum \Expc{}{X_i} + \varepsilon} \leq e^{-2\varepsilon^2 n},\;\;\Prob{}{\frac{1}{n}\sum X_i \leq \frac{1}{n}\sum \Expc{}{X_i} - \varepsilon} \leq e^{-2\varepsilon^2 n}.$$
The same inequalities hold for random variables taking values in $[0,1]$ (Hoeffding bound). The multplicative version of the Chernoff bound states:
$$\Prob{}{\sum X_i \geq (1+\varepsilon)\cdot\sum \Expc{}{X_i}} \leq e^{-\varepsilon^2 \sum \Expc{}{X_i}/3},\;\;\Prob{}{\sum X_i \leq (1 - \varepsilon)\cdot\sum \Expc{}{X_i}} \leq e^{-\varepsilon^2 \sum \Expc{}{X_i}/2}.$$
\end{lemma}

When we say that a quantity is {\em exponentially small in $k$} we mean that it is of the form $2^{-\Omega(k)}$. We use $\exp(-n)$ to mean $e^{-n}$.

\subsection{Non-Uniform and Randomized Algorithms}

\begin{definition}[Non-uniform algorithm]
A non-uniform algorithm that runs in time $t(n)$ is given by a sequence $\set{C_n}$ of Boolean circuits, where for every $n\geq 1$, the circuit $C_n$ gets an input of size $n$ and satisfies $\card{C_n}\leq t(n)$.

Alternatively, 
a non-uniform algorithm that runs in time $t(n)$ on input of size $n$ is given an advice string $a=a(n)$ of size at most $t(n)$ (note that $a(n)$ depends on $n$ but not on the input!). The algorithm runs in time $t(n)$ given the input and the advice.
\end{definition}

The class of all languages that have non-uniform polynomial time algorithms is called P$/\poly$.

There are two main types of randomized algorithms: the strongest are Las Vegas algorithms that may not return a correct output with small probability, but would report their failure. Atlantic City algorithms simply return an incorrect output a small fraction of the time.

\begin{definition}[Las Vegas algorithm]
A {\em Las Vegas} algorithm that runs in time $t(n)$ on input of size $n$ is a randomized algorithm that always runs in time at most $t(n)$, but may, with a small probability return $\bot$. In any other case, the algorithm returns a correct output.
\end{definition}
The probability that a Las Vegas algorithm returns $\bot$ is called its {\em error probability}. In any other case we say that the algorithm succeeds. 
\begin{definition}[Atlantic City algorithm]
An {\em Atlantic City} algorithm that runs in time $t(n)$ on input of size $n$ is a randomized algorithm that always runs in time at most $t(n)$, but may, with a small probability, return an incorrect output. 
\end{definition}
The probability that an Atlantic City algorithm returns an incorrect output is called its {\em error probability}. In any other case we say that the algorithm succeeds.

Note that a Las Vegas algorithm is a special case of Atlantic City algorithms. Atlantic City algorithms that solve decision problems return the same output the majority of the time. For search problems we have the following notion:

\begin{definition}[Pseudo-deterministic algorithm, \cite{GatG}]
A {\em Pseudo-deterministic} algorithm is an Atlantic City algorithm that returns the same output except with a small probability, called its {\em error probability}.
\end{definition}

\section{Derandomization by Oblivious Verification}\label{s:method}

In this section we develop a technique for converting randomized algorithms to deterministic non-uniform algorithms. The derandomization technique is based on the notion of ``oblivious verifiers'', which are verifiers that deterministically test the randomness of an algorithm while accessing only a sketch (compressed version) of the input to the algorithm.
If the verifier accepts, the algorithm necessarily succeeds on the input and the randomness. In contrast, the verifier is allowed to reject randomness strings on which the randomized algorithm works correctly, as long as it does not do so for too many randomness strings. \begin{definition}[Oblivious verifier] Suppose that $A$ is a randomized algorithm that on input $x\in\set{0,1}^N$ uses $p(N)$ random bits. Let $s:\N\to\N$ and $\varepsilon:\N\to [0,1]$.
An $(s,\varepsilon)$-{\em oblivious verifier} for $A$ is a deterministic procedure that gets as input $N$, a sketch $\hat{x}\in\set{0,1}^{s(N)}$ and $r\in\set{0,1}^{p(N)}$, either accepts or rejects, and satisfies the following:
\begin{itemize}
\item Every $x\in\set{0,1}^N$ has a sketch $\hat{x}\in\set{0,1}^{s(N)}$.
\item For every $x\in\set{0,1}^N$ and its sketch $\hat{x}\in\set{0,1}^{s(N)}$, for every $r\in\set{0,1}^{p(N)}$, if the verifier accepts on input $\hat{x}$ and $r$, then $A$ succeeds on $x$ and $r$.
\item For every $x\in\set{0,1}^N$ and its sketch $\hat{x}\in\set{0,1}^{s(N)}$,  the probability over $r\in\set{0,1}^{p(N)}$ that the verifier rejects is at most $\varepsilon(N)$.
\end{itemize}
\end{definition} 
Note that $\varepsilon$ of the oblivious verifier may be somewhat larger than the error probability of the algorithm $A$, but hopefully not much larger.
We do not limit the run-time of the verifier, but the verifier has to be deterministic. Indeed, the oblivious verifiers we design run in deterministic exponential time. We do not limit the time for computing the sketch $\hat{x}$ from the input $x$ either. Indeed, we use the probabilistic method in the design of our sketches. Crucially, the sketch depends on the input $x$, but is independent of $r$.

%\begin{figure}[h]\centering
%\includegraphics{union-bound-figure}
%\begin{caption}{The existence of an oblivious verifier effectively shrinks the %number of inputs, since inputs with the same sketch can be bundled %together.}\label{f:union-bound}
%\end{caption}
%\end{figure}

Our derandomization theorem shows how to transform a randomized algorithm with an oblivious verifier into a deterministic (non-uniform) algorithm whose run-time is not much larger than the run-time of the randomized algorithm. Its idea is as follows. An oblivious verifier allows us to partition the inputs so inputs with the same sketch are bundled together, and the number of inputs effectively shrinks. This allows us to apply a union bound, just like in Adleman's proof~\cite{Adleman}, but over many fewer inputs, to argue that there must exist a randomness string for (a suitable repetition of) the randomized algorithm that works for all inputs. 

%% exists randomness good for all inputs whp

\begin{theorem}[Derandomizing by verifying from a sketch]\label{t:derandomization} For every $t\geq 1$, if a problem has a Las Vegas algorithm that runs in time $T$ and a corresponding $(s,\varepsilon)$-oblivious verifier for $\varepsilon< 2^{-s/t}$, then the problem has a non-uniform deterministic algorithm that runs in time $T\cdot t$ and always outputs the correct answer.
\end{theorem}
\begin{proof} Consider the randomized algorithm that runs the given randomized algorithm on its input for $t$ times independently, and succeeds if any of the runs succeeds. Its run-time is $T\cdot t$.
For any input, the probability that the oblivious verifier rejects all of the $t$ runs is less than $(2^{-s/t})^t= 2^{-s}$.
By a union bound over the $2^s$ possible sketches, the probability that the oblivious verifier rejects for any of the sketches is less than $2^s\cdot 2^{-s}=1$. Hence, there exists a randomness string that the oblivious verifier accepts no matter what the sketch is. On this randomness string the algorithm has to be correct no matter what the input is. 
The deterministic non-uniform algorithm invokes the repeated randomized algorithm on this randomness string.
\end{proof}
Adleman's theorem can be seen as a special case of Theorem~\ref{t:derandomization}, in which the sketch size is the trivial $s(N)=N$, the oblivious verifier runs the algorithm on the input and randomness and accepts if the algorithm succeeds, and the randomized algorithm has error probability $\varepsilon < 2^{-N/t}$.

The reason that we require that the algorithm is a Las Vegas algorithm in Theorem~\ref{t:derandomization} is that it allows us to repeat the algorithm and combine the answers from all invocations. Combining is possible by other means as well. E.g., for randomized algorithms that solve decision problems or for pseudo-deterministic algorithms (algorithms that typically return the same answer) one can combine by taking majority. For algorithms that return a list, one can combine the lists.

The derandomization technique assumes that the error probability of the algorithm is sufficiently low. To complement it, in Section~\ref{s:biased} we develop an amplification technique to decrease the error probability. Interestingly, our applications are such that the error probability can be decreased without a substantial slowdown to a point at which our derandomization technique kicks in, but we do not know how to decrease the error probability sufficiently for Adleman's original proof to work without slowing down the algorithm significantly.

\section{Amplification by Finding a Biased Coin}\label{s:biased}

In this section we develop a technique that will allow us to significantly decrease the error probability of randomized algorithms without substantially slowing down the algorithms. The technique works by testing the random choices made by the algorithm and quickly discarding undesirable choices. It requires the ability to quickly estimate the desirability of random choices.
The technique is based on a solution to the following problem.

\begin{definition}[Biased coin problem] Let $0<\eta,\zeta<1$.
In the biased coin problem one has a source of coins. Each coin has a bias, which is the probability that the coin falls on ``heads''. 
The bias of a coin is unknown, and one can only toss coins and observe the outcome. 
It is known that at least $2/3$ fraction\begin{footnote}{$2/3$ can be replaced with any constant larger than $0$.}\end{footnote} of the coins have bias at least $1-\eta$.
Given $n\geq 1$, the task is to find a coin of bias at least $1-\eta - \zeta$ with probability at least $1-\exp(-n)$ using as few coin tosses as possible.
\end{definition}
A similar problem was studied in the setup of multi-armed bandit problems~\cite{EMM,MT}, however in that setup there might be only one coin with large bias, as opposed to a constant fraction of coins as in our setup. In the former setup, many more coin tosses might be needed (an $\Omega(n^2/\zeta^2)$ lower bound is proved in~\cite{MT}).

The analogy between the biased coin problem and amplification is as follows: a coin corresponds to a random choice of the algorithm. Its bias corresponds to how desirable the random choice is. The assumption is that a constant fraction of the random choices are very desirable. The task is to find one desirable random choice with a very high probability. Tossing a coin corresponds to testing the random choice. The coin falls on heads in proportion to the quality of the random choice.   

Interestingly, if we knew that all coins have bias either at least $1-\eta$ or at most $1-\eta-\zeta$, it would have been possible to solve the biased coin problem using only $O(n/\zeta^2)$ coin tosses. The algorithm is described in Figure~\ref{alg:coin0}. It tosses a random coin a small number of times and expects to witness about $1-\eta$ fraction heads. If so, it doubles the number of tosses, and tries again, until its confidence in the bias is sufficiently large. If the fraction of heads is too small, it restarts with a new coin. The algorithm has two parameters $i_0$ that determines the initial number of tosses and $i_f$ that determines the final number of tosses. 

The probability that the algorithm restarts at the $i$'th phase is exponentially small in $\zeta^2 k$ for $k=2^{i}$: either the coin had bias at least $1-\eta$, and then there's an exponentially small probability in $\zeta^2 k$ that there were less than $(1-\eta - \zeta/2) k$ heads, or the coin had bias at most $1-\eta-\zeta$, and then there is probability exponentially small in $\zeta^2 k$ that the coin had at least $1-\eta-\zeta/2$ fraction heads in all the previous phases (whereas if this is phase $i=i_0$, then the probability that a coin with bias less than $1-\eta$ was picked in this case is constant, i.e., exponentially small in $\zeta^2 k$). Moreover, the number of coin tosses up to this step is at most $2k$. Hence, we maintain a linear relation (up to $\zeta^2$ factor) between the number of coin tosses and the exponent of the probability. To get the error probability down to $\exp(-n)$ we only need $O(n/\zeta^2)$ coin tosses.

\begin{algbox}{6.5in}{coin0}{An algorithm for finding a coin of bias at least $1-\eta-\zeta$ when all the coins either have bias at least $1-\eta$ or at most $1-\eta-\zeta$. The algorithm uses $O(n/\zeta^2)$ coin tosses and achieves error probability $\exp(-n)$.}
	\Procname{\textsc{Find-Biased-Coin-Given-Gap}$(n,\eta,\zeta)$}
	\li Set $i_0= \log(1/\zeta^2)+ \Theta(1)$; $i_f= \log(n/\zeta^2)+\Theta(1)$ (constants picked appropriately).
	\li Pick a coin at random.
	\li \For $i=i_0,i_0 + 1,\ldots,i_f$ \Do
		\li Toss the coin for $k=2^i$ times.
		\li If the fraction of heads is less than $1-\eta-\zeta/2$, restart.
	\End
	\li \Return coin.
\end{algbox}

Counter-intuitively, adding coins of bias between $1-\eta-\zeta$ and $1-\eta$ -- all acceptable outcomes of the algorithm -- derails the algorithm we outlined above, as well as other algorithms.
If one fixes a threshold like $1-\eta-\zeta/2$ for the fraction of heads one expects to witness, there might be a coin whose bias is around the threshold. We might toss this coin a lot and then decide to restart with a new coin. One can also consider a competition-style algorithm like the ones studied in~\cite{EMM,MT} when one tries several coins each time, keeping the ones that fall on heads most often. 
Such a algorithm may require $\Omega(n^2/\zeta^2)$ coin tosses, since coins can lose any short competition to coins with slightly smaller bias; then, such coins can lose to coins with slightly smaller bias, and so on, until we may end up with a coin of bias smaller than $1-\eta-\zeta$.    

There is, however, a algorithm that uses only $\tilde{O}(n/\zeta^2)$ coin tosses. This algorithm decreases the threshold for the fraction of heads one expects to witness with respect to the number of coin tosses one already made for this coin. If the coin was already tossed a lot, the deviation of the number of heads from $1-\eta$ would have to be large for us to decide to restart with a new coin. 
The algorithm is described in Figure~\ref{alg:coin}.

\begin{algbox}{6.5in}{coin}{An algorithm for finding a coin of bias at least $1-\eta-\zeta$ using $\tilde{O}(n/\zeta^2)$ coin tosses. The error probability is exponentially small in $n$.}
	\Procname{\textsc{Find-Biased-Coin}$(n,\eta,\zeta)$}
	\li Set $i_0= \log(1/\zeta^2)+\Theta(1); \;\; i_f= \log(n/\zeta^2)+\Theta(\log\log(n/\zeta));\;\; 
\beta = \zeta/i_f$.
	\li Pick a coin at random.
	\li \For $i=i_0,i_0+1,\ldots,i_f$ \Do
		\li Toss the coin for $k=2^i$ times.
		\li If the fraction of heads is less than $1-\eta - i\beta$, restart.
	\End
	\li \Return coin.
\end{algbox}

Note that the deviation parameter $\beta$ is picked so $1-\eta-i\beta\geq 1-\eta-\zeta$ for all $i\leq i_f$.

\begin{lemma} Within $O((n/\zeta^2)\log^2 (n/\zeta)) = \tilde{O}(n/\zeta^2)$ coin tosses, \textsc{Find-Biased-Coin} outputs a coin of bias at least $1-\eta-\zeta$ except with probability $\exp(-n)$.
\end{lemma}
\begin{proof}
Suppose that the algorithm restarts at phase $i$. The number of coin tosses made by this point since the previous restart (if any) is $2^{i_0} + 2^{i_0+1} + \ldots + 2^i \leq 2^{i+1}$. Moreover, if the coin had bias smaller than $1-\eta - i\beta + \beta/2$, then, if $i> i_0$, by a Chernoff bound, the probability the coin passed the previous test, where it was supposed to have at least $1-\eta - (i-1)\beta$ fraction of heads, is at most $\exp(-\beta^2 2^{i-2})$. If $i=i_0$, there is probability less than $1/3$ that the coin was picked. 
If the coin had bias at least $1-\eta - i\beta + \beta/2$, then by the Chernoff bound, the probability it failed the current test, where it is supposed to have at least $1-\eta - i\beta$ fraction of heads, is at most $\exp(-\beta^2 2^{i-1})$.
In any case, the ratio between the number of coin tosses and the exponent of the probability is $O(1/\beta^2)$. 
The value of $i_f$ is chosen so that the error probability in the last iteration is $\exp(-n)$.
By the choice of $\beta$, the coin tosses to exponent ratio is $O((1/\zeta^2)\log^2 (n/\zeta))$. Therefore, the number of coin tosses one needs in order to reach error probability $\exp(-n)$ is $O((n/\zeta^2)\log^2 (n/\zeta))$. 
\end{proof}

\subsection{Extensions}\label{s:approx-coin} 

In the sequel, we'll use the biased coin algorithm in a more general setting, and in this section we develop the appropriate machinery. In the general setting coins are divided into groups, and rather than directly tossing coins we simulate tossing. The simulation may fail or may produce results that are inconsistent with the true bias of the coin. Some of the coins may be faulty, and their tossing may fail arbitrarily. 
For other coins, the probability that tossing fails is small. For any coin, the probability that the coin toss does not fail and is inconsistent with the true bias is small. The coins are partitioned into groups of size $g$ each. The bias of a group is the maximum bias among non-faulty coins in the group, and is $0$ if all the coins in the group are faulty.
At least $2/3$ fraction of the groups have bias at least $1-\eta$.
The task is to find a group of coins of bias at least $1-\eta-\zeta$.

The formal requirements from a simulated coin toss are as follows:
\begin{definition}\label{d:simulate-toss}
Simulated coin tosses satisfy the following:
\begin{itemize}
\item For any coin, when tossing the coin $k$ times, there is exponentially small probability in $\beta^2 k$ for the following event: the tosses do not fail, yet the fraction of tosses that fall on heads deviates from the true bias by more than an additive $\beta/4$ for $\beta$ as in Figure~\ref{alg:coin-extended}.

\item For non-faulty coins, the probability that tossing the coin fails is exponentially small in $\beta^2 k$.

\end{itemize}
\end{definition}
Note that the simulation has to be very accurate, since the deviation $\beta/4$ is sub-constant. We describe a modified biased coin algorithm in Figure~\ref{alg:coin-extended}. 

\begin{algbox}{6.5in}{coin-extended}{An algorithm for finding a group of coins of bias at least $1-\eta-\zeta$, where the coins are partitioned into groups of size $g$ each. The error probability is exponentially small in $n$.}
	\Procname{\textsc{Find-Biased-Coin-in-Group}$(n,g,\eta,\zeta)$}
	\li Set $i_f= \log((n+\log g)/\zeta^2)+\Theta(\log\log((n+\log g)/\zeta))$;\;\; $\beta = \zeta/i_f$.
	\li Set $i_0= \log(\log g/\beta^2)+\Theta(1)$, where the constant term is sufficiently large.
	\li Pick a group of coins at random.
	\li \For $i=i_0,i_0+1,\ldots,i_f$ \Do
		\li Simulate tossing each one of the $g$ coins in the group for $k=2^i$ times.
		\li If the maximum fraction of heads per coin is less than $1-\eta - i\beta$, restart.
	\End
	\li \Return group of coins.
\end{algbox}

\begin{lemma}[Generalized biased coin]\label{l:simulated-coin}
If \textsc{Find-Biased-Coin-in-Group} restarts at a certain phase, then either in this phase or in the previous, the reported fraction of heads deviates by more than $\beta/2$ from the true bias for one of the non-faulty coins in the group, or it is the first phase and a group of bias at most $1-\eta - \beta/2$ was picked. 

As a result, within $O((ng\log g/\zeta^2)\log^2 ((n+\log g)/\zeta))=\tilde{O}(ng/\zeta^2)$ coin tosses \textsc{Find-Biased-Coin-in-Group} outputs a coin of bias at least $1-\eta-\zeta$ except with probability $\exp(-n)$.   
\end{lemma}
\begin{proof}
Suppose that the algorithm restarts at phase $i$. The number of coin tosses made by this point is $g\cdot(2^{i_0} + 2^{i_0+1} + \ldots + 2^i) \leq g\cdot 2^{i+1}$. 

Suppose that the group had bias smaller than $1-\eta - i\beta + \beta/2$. If $i> i_0$, the probability that the coins passed the previous test, where at least one of them was supposed to have at least $1-\eta - (i-1)\beta$ fraction of heads, is $g\cdot(\exp(-\beta^2 2^{i-4}) + \exp(-\Omega(\beta^2 2^{i-1})))$ (where $\beta/4$ of the deviation and $\exp(-\Omega(\beta^2 2^{i-1}))$ of the error probability can be attributed to the simulation). Note that this probability is exponentially small in $\beta^2 2^i$ when $\log g$ is sufficiently smaller than $\beta^2 2^{i_0-1}$ (here we use the choice of $i_0$).
If $i=i_0$, the probability that a group of bias smaller than $1-\eta$ was picked is less than $1/3$.

On the other hand, if the group has bias at least $1-\eta - i\beta + \beta/2$, then the probability it failed the current test, where one of the coins is supposed to have at least $1-\eta - i\beta$ fraction of heads, is at most $\exp(-\beta^2 2^{i-3})+\exp(-\Omega(\beta^2 2^{i}))$ (again, $\beta/4$ of the deviation and $\exp(-\Omega(\beta^2 2^{i}))$ of the error probability can be attributed to the simulation).

In any case, the ratio between the number of coin tosses and the exponent of the probability is $O(g\log g/\beta^2)$. 
The value of $i_f$ is set so the error probability in the last iteration is $\exp(-n)$.
By the choice of $\beta$, the coin tosses to exponent ratio is $O((g\log g/\zeta^2)\log^2 ((n+\log g)/\zeta))$. Therefore, the number of coin tosses one needs in order to reach error probability $\exp(-n)$ is $O((ng\log g/\zeta^2)\log^2 ((n+\log g)/\zeta))= \tilde{O}(ng/\zeta^2)$.
\end{proof}

%A algorithm described in Figure~\ref{alg:approx-coin}, which is nearly identical %to the algorithm of Figure~\ref{alg:coin}, continues to work in this more %general setting.

%\begin{algbox}{6.5in}{approx-coin}{A algorithm for finding a coin of bias at %least $1-c\eta$ using $\tilde{O}(n)$ coin tosses. The error probability of the %algorithm is exponentially small in $n$, and it works even if tossing a coin $k$ %times may be truthful to the bias up to additive error $\beta$, and except with %probability $\exp(-\Omega(k))$.}
%	\Procname{\textsc{Find-Biased-Coin-With-Approximate-Tosses}}
%	\li Set $p = \log n/\log\log n\geq 2$ and $\beta = \eta\log p /\log n$. 
%	\li Pick a coin at random.
%	\li \For $i=1,2,\ldots,\log n/\log p$ \Do
%		\li Toss the coin for $k=p^i$ times.
%		\li If the fraction of heads is less than $1-\eta - 2i\beta$, restart.
%	\End
%	\li \Return coin.
%\end{algbox}

\section{Max Cut on Dense Graphs}\label{s:max-cut}

In this section we show the application to \textsc{Max-Cut} on dense graphs. This is our simplest example. It relies on the biased coin algorithm, and does not require any sketches.

\subsection{A Simple Randomized Algorithm}

First we describe a simple randomized algorithm for dense \textsc{Max-Cut} based on the sampling idea of Fernandez de la Vega~\cite{Vega} and Arora, Karger and Karpinski~\cite{AKK95}. We remark that Mathieu and Schudy~\cite{MS} have similar, but more efficient, randomized algorithms, however, for the sake of simplicity, we stick to the simplest algorithm with the easiest analysis.

The main idea of the algorithm is as follows. 
We sample a small $S\subseteq V$ and enumerate over all possible $S$-cuts $H\subseteq S$. Each $S$-cut induces a cut $C_{S,H}\subseteq V$ as follows.

\begin{definition}[Induced cut]
Let $G=(V,E)$. Let $S\subseteq V$ and $H\subseteq S$. We define $C_{S,H}\subseteq V$ as follows:
for every $v\in V$ let $v\in C_{S,H}$ if the fraction of edges $e=(v,s)\in E$ with $s\in S-H$ is larger than the fraction of edges $e=(v,s)\in E$ with $s\in H$.
\end{definition}

We will argue below that if there is a cut in $G$ with value at least $1-\varepsilon$ and $H$ is the restriction of that cut to $S$, then the induced cut is likely to approximately achieve the optimal value.
Note that we rely on density when we hope that the edges that touch the small set $S$ span most of the vertices in the graph.

\begin{lemma}[Sampling]\label{l:sampling}
Let $G=(V,E)$ be a regular $\gamma$-dense graph that has a cut of value at least $1-\varepsilon$ for $\varepsilon < 1/4$. Then for $\zeta < 1/4-\varepsilon$ and for a uniform $S\subseteq V$, $$\card{S}=\max\set{\ceil{\log(2/\zeta^2)/\zeta^2} ,\ceil{2\log(2/\zeta^2)/\gamma^2}},$$ with probability at least $1-\zeta$, there exists $H\subseteq S$ such that the value of the cut $C_{S,H}$ is at least $1 - \varepsilon - 10\zeta$. 
\end{lemma}
\begin{proof}
Let $C^*$ be the optimal cut in $G$.
Let $H\subseteq S$ be the restriction of $C^*$ to $S$. Denote by $V'$ the set of all $v\in V$ such that at least $1/2 + \zeta$ fraction of the edges that touch $v$ contribute to the value of $C^*$. 
Note that $\card{V'}\geq (1-4\zeta)\card{V}$. By
$\gamma$-density and regularity, the degree of all vertices is $\gamma\card{V}$.
By a Chernoff bound, except with probability $\zeta^2/2$ over $S$, at least $(\gamma/2)\card{S}$ of the vertices in $S$ are neighbors of $v$. The sample of $v$'s neighbors is uniform and hence by another Chernoff bound, except with probability $\zeta^2/2$ over $S$, the vertex $v$ is assigned by $C_{S,H}$ to the same side as $C^*$ assigns it.
Therefore, except with probability $\zeta$ over the random choice of $S$, at least $1-\zeta$ fraction of the vertices $v\in V'$ are assigned by $C_{S,H}$ the same as $C^*$. This means that at least $1-5\zeta$ fraction of the vertices $v\in V$ are assigned by $C_{S,H}$ the same as $C^*$. 
Therefore, the fraction of edges that: (i) contribute to the value of $C^*$, and (ii) have both their endpoints assigned by $C_{S,H}$ the same as $C^*$, is at least $1-\varepsilon - 2\cdot 5\zeta = 1-\varepsilon-10\zeta$.
\end{proof}

\subsection{A Randomized Algorithm With Exponentially Small Error Probability}
 
We describe an analogy between finding a cut of high value and finding a biased coin. 
We think of sampling $S\subseteq V$ as picking a group of coins, and picking $H\subseteq S$ as picking a coin in the group. The bias of the coin is the value of the cut $C_{S,H}$. Therefore a biased coin directly corresponds to a desirable cut.
One tosses a coin by picking an edge $(u,v)\in E$ uniformly at random, computing whether $u\in C_{S,H}$ and whether $v\in C_{S,H}$, and checking whether the edge contributes to the value of the cut. Note that checking whether a vertex belongs to $C_{S,H}$ is computed in time $\card{S}$. The coin toss algorithm is described in Figure~\ref{alg:toss-cut-coin}. 
The algorithm based on finding a biased coin is described in Figure~\ref{alg:amplified-cut}.

\begin{algbox}{6.5in}{toss-cut-coin}{A coin toss picks an edge at random and checks whether it contributes to the value of the cut $C_{S,H}$. }
	\Procname{\textsc{Max-Cut-Toss-Coin}$(G=(V,E),S,H)$}
	\li Pick $e=(u,v)\in E$ uniformly at random.
	\li \Return ``heads'' iff $u\in C_{S,H}$ and $v\notin C_{S,H}$ or vice versa.
\end{algbox}

\begin{algbox}{6.5in}{amplified-cut}{An algorithm for finding a cut of value $1-\varepsilon-O(\zeta)$ in a regular $\gamma$-dense graph that has a cut of value $1-\varepsilon$. The error probability of the algorithm is exponentially small in $\card{V}^2$.}
	\Procname{\textsc{Find-Cut}($G=(V,E),\varepsilon,\zeta$)}
	\li Set $s=\max\set{\ceil{\log(2/\zeta^2)/\zeta^2} ,\ceil{2\log(2/\zeta^2)/\gamma^2}}$, where $\gamma$ is the density of $G$.
	\li Set $i_f= \log((\card{V}^2+s)/\zeta^2)+\Theta(\log\log((\card{V}+s)/\zeta))$;\;\; $\beta = \zeta/i_f$.
	\li Set $i_0= \log(s/\beta^2)+\Theta(1)$.
	\li Sample $S\subseteq V$, $\card{S}=s$.
	\li \For $i=i_0,i_0+1,\ldots,i_f$ \Do
		\li \For all $H\subseteq V$ \Do
					\li Invoke \textsc{Max-Cut-Toss-Coin}$(G,S,H)$ for $k=2^i$ times.
		\End
		\li If the fraction of heads is less than $1- \varepsilon - 10\zeta - i\beta$ for all $H$, restart.
	\End
	\li \Return cut $C_{S,H}$ with value at least $1-\varepsilon - 11\zeta$ if exists.
\end{algbox}

This proves Theorem~\ref{t:max-cut}, which is repeated below for convenience. Note that for a sufficiently small error probability exponentially small in $\card{V}^2$ it follows that there exists a randomness string on which the algorithm succeeds, no matter what the input is. 
\begin{theorem}
There is a Las Vegas algorithm that given a $\gamma$-dense graph $G$ that has a cut of value at least $1-\varepsilon$ for $\varepsilon < 1/4$, and given $\zeta < 1/4-\varepsilon$, finds a cut of value at least $1-\varepsilon - O(\zeta)$, except with probability exponentially small in $\card{V}^2$. The algorithm runs in time $\tilde{O}(\card{V}^2(1/\zeta)^{O(1/\gamma^2 + 1/\zeta^2)})$. It also implies a non-uniform deterministic algorithm with the same run-time.
\end{theorem}

\section{Approximate Clique}\label{s:clique}

\subsection{An Algorithm With Constant Error Probability}\label{s:dense-randomized}

In this section we describe a randomized algorithm with constant error probability for finding an approximate clique in a graph that has a large clique. The algorithm is a simplified version of an algorithm and analysis by
Goldreich, Goldwasser and Ron~\cite{GGR}. We rely on the algorithm and the analysis when we design a randomized algorithm with error probability $\exp(-\Omega(\card{V}))$ and again when we design a deterministic algorithm.

The main idea of the algorithm is as follows. 
We first find a small random subset $U'$ of the large clique $C$ by sampling vertices $U$ from $V$ and enumerating over all possibilities for $C\cap U$.
The intuition is that now we would like to find other vertices that are part of the large clique $C$. A natural test for whether a vertex $v\in V$ is in the clique is whether $v$ is connected to all the vertices in $U'$. This, however, is not a sound test, since the clique might have many vertices that neighbor it but do not neighbor one another. A better test checks whether $v$ neighbors all of $U'$, as well as many of the vertices that neighbor all of $U'$. Vertices that neighbor all of $U'$ are likely to neighbor almost all of the clique. 
%We know that the vertices in the clique $C$ are likely to have large degrees in %$\Gamma(U')$. The set found by the algorithm can only have larger density.

The algorithm is described in Figure~\ref{alg:random}. It picks $U\subseteq V$ at random, considers all possible sub-cliques $U'\subseteq U$, $\card{U'}\geq (\rho/2)\card{U}$, computes $\Gamma(U')$ the set of vertices that neighbor all of $U'$, computes for every vertex in $\Gamma(U')$ the fraction of vertices in $\Gamma(U')$ that neighbor it, and considers $S_{U'}$ the set of $\rho\card{V}$ vertices in $\Gamma(U')$ with largest fractions of neighbors. The algorithm outputs a sufficiently dense set among all sets $S_{U'}$, if such exists.

\begin{algbox}{6.5in}{random}{
Randomized algorithm with constant error probability for finding an approximate clique.
}
	\Procname{\textsc{Find-Approximate-Clique-Constant-Error}$(G=(V,E),\rho,\varepsilon)$}
	\li Sample $U\subseteq V$, $\card{U}=\ceil{k_0/\rho}$, for $k_0=100/\varepsilon^2$. 
	\li \For all sub-cliques $U'\subseteq U$, $\card{U'}\geq (\rho/2)\card{U}$, \Do
		\li Compute $\Gamma(U')$ the set of vertices that neighbor all of $U'$.
		\li For each $v\in \Gamma(U')$ compute the fraction ${f}_v$ of vertices in $\Gamma(U')$ that neighbor $v$.
		\li Let $S_{U'}\subseteq \Gamma(U')$ contain the $\rho\card{V}$ vertices with largest $f_v$.
	\End
	\li \Return set $S_{U'}$ of density at least $1-2\varepsilon/\rho$ if such exists.
\end{algbox}
The algorithm runs in time $\exp(k_0/\rho)\cdot O(\card{V}^2)$. Next we analyze the probability it is correct. 
By a Chernoff bound, we have $\card{U\cap C}\geq (\rho-\varepsilon)\card{U}$, except with probability $1/10$.
Pick a uniformly random order on the vertices.
Let us focus on the event $\card{U\cap C}\geq (\rho-\varepsilon)\card{U}$ and $U'$ that is the first $(\rho/2)\card{U}$ elements in $U\cap C$ according to the random order. Note that the elements of $U'$ are uniformly and independently distributed in $C$. Let $\Gamma(U')\subseteq V$ contain all the vertices that neighbor all of $U'$.

\begin{lemma}\label{l:U'-approx} With probability $1-e^{-25/\varepsilon}$ over the choice of $U'$, the fraction of $v\in \Gamma(U')$ that neighbor less than $1-\varepsilon$ fraction of $C$ is at most $e^{-25/\varepsilon}$.
\end{lemma}
\begin{proof}
Consider $v\in V$ that has less than $1-\varepsilon$ neighbors in $C$. 
For $v$ to be in $\Gamma(U')$ the set $U'$ must miss all of the non-neighbors of $v$. Since $U'$ is a uniform sample of $C$, this happens with probability $(1-\varepsilon)^{\card{U'}}\leq e^{-50/\varepsilon}$. The lemma follows.
\end{proof}  

Let us focus on $U'$ for which the fraction of $v\in \Gamma(U')$ that neighbor less than $1-\varepsilon$ fraction of $C$ is at most $e^{-25/\varepsilon}$. Lemma~\ref{l:U'-approx} guarantess that such a $U'$, which we call {\em good}, is picked with constant probability. Next we show that an average vertex in $C$ neighbors most of $\Gamma(U')$. %(Lemma~\ref{l:C-degree}).
%, and that one obtains a good approximation to the density between any set %$A\subseteq V$ and $\Gamma(U')$ (Lemma~\ref{l:density-approx}). 

\begin{lemma}[Density for $C$]\label{l:C-degree} For good $U'$, the average number of neighbors a vertex $c\in C$ has in $\Gamma(U')$ is at least $(1-2\varepsilon)\cdot\card{\Gamma(U')}$.
\end{lemma}
\begin{proof}
Since $U'$ is good, more than $1-e^{-25/\varepsilon}$ fraction of $\Gamma(U')$ neighbor at least $1-\varepsilon$ fraction of $C$. Hence, the average fraction of $\Gamma(U')$ neighbors a uniform vertex in $C$ has is at least $1-2\varepsilon$ (using $e^{-25/\varepsilon}\leq \varepsilon$). 
\end{proof}

We can now prove the correctness of \textsc{Find-Approximate-Clique-Constant-Error}.

\begin{lemma}\label{l:clique-const-err}
With probability at least $1-e^{-25/\varepsilon}$, \textsc{Find-Approximate-Clique-Constant-Error}, when invoked on $0<\rho,\varepsilon<1$, and a graph $G=(V,E)$ with a clique on $\rho\card{V}$ vertices, picks $S_{U'}$ such that $({1}/{\card{S_{U'}}})\cdot\sum_{v\in S_{U'}} f_v \geq 1-2\varepsilon$, and returns a set of density at least $1-2\varepsilon/\rho$.
\end{lemma}
\begin{proof}
For good $U'$, by Lemma~\ref{l:C-degree}, $({1}/{\card{C}})\sum_{v\in C} f_v\geq 1-2\varepsilon$. Since $S_{U'}$ takes the $\rho\card{V}$ vertices with largest $f_v$ and $\card{C}\geq\rho\card{V}$, we have $({1}/{\card{S_{U'}}})\cdot\sum_{v\in S_{U'}} f_v \geq 1-2\varepsilon$. Therefore, the density within $S_{U'}$ is at least $1-2\varepsilon/\rho$, and so is the density of the set returned by the algorithm.
\end{proof}

Next we show how to transform the randomized algorithm with constant error probability from Section~\ref{s:dense-randomized} into an algorithm with error probability that is exponentially small in $\card{V}$ {\em without increasing the run-time by more than poly-logarithmic factors}. The algorithm applies the biased coin algorithm from Section~\ref{s:biased}.

\subsection{Finding an Approximate Clique as Finding a Biased Coin}\label{s:clique-amplified}

%The algorithm that finds a dense set with error probability exponentially small %in $\card{V}$ uses the algorithm with constant error probability from %Section~\ref{s:dense-randomized} and the algorithm for finding a biased coin %from Section~\ref{s:biased}. 
The analogy between finding a biased coin and finding an approximate clique is as follows: Picking $U$ picks a group of coins. There is a coin for every $U'\subseteq U$, $\card{U'}\geq (\rho/2)\card{U}$. The coin is faulty if $\card{\Gamma(U')}< \rho\card{V}$.
%We'll define the bias of a faulty coin as $0$.
A coin corresponds to the set $S_{U'}$ of the $\rho\card{V}$ vertices in $\Gamma(U')$ with largest number of neighbors in $\Gamma(U')$ (when the coin is faulty, pad the set with dummy vertices with $0$ neighbors).  
The bias of the coin $bias_{U'}$ is the expectation, over the choice of a random vertex $v\in S_{U'}$, of the fraction of vertices in $\Gamma(U')$ that neighbor $v$.
With at least $2/3$ probability, one of the coins in the group -- the one associated with a good $U'$ in the sense of Section~\ref{s:dense-randomized} -- has $bias_{U'}\geq 1-2\varepsilon$. Moreover, any $U'$ with $bias_{U'}\geq 1-c\varepsilon$ corresponds to a set of density at least $1-c\varepsilon/\rho$.

Had we found the vertices in each $S_{U'}$, we could have tossed a coin by picking a vertex at random from $S_{U'}$ and a vertex at random from $\Gamma(U')$ and letting the coin fall on heads if there is an edge between the two vertices. Unfortunately, finding the vertices in $S_{U'}$ may take
$\Omega(\card{V}^2)$ time, so we refrain from doing that. We settle for a simulated toss -- where with high probability the coin falls on heads with probability close to its bias. 
In Section~\ref{s:approx-coin} we extended the biased coin algorithm to simulated tosses. In Figure~\ref{alg:coin-toss} we describe the algorithm for tossing a coin enough times so the probability of $\gamma$-deviation from the true bias is exponentially small in $k$ (the number of coin tosses is implicit). 
The algorithm runs in time $O(k\card{V}\card{U'}\poly(1/\rho,1/\gamma))$.

\begin{algbox}{6.5in}{coin-toss}{
An algorithm for tossing the coin associated with $U'$, where the coin falls on heads with probability $\Theta(\gamma)$-close to its bias except with probability exponentially small in $k$.
}
	\Procname{\textsc{Clique-Coin-Toss}$(G=(V,E),U',\rho,k,\gamma)$}
		\li Compute $\Gamma(U')$.
		\li \If $\card{\Gamma(U')}< \rho\card{V}$ \Then
			\li Fail.
		\End
		\li Sample $V'\subseteq V$, $\card{V'}=\ceil{k/(\rho\gamma^2)}$.
		\li For all $v\in V'$ compute the fraction $f_v$ of $\Gamma(U')$ vertices that neighbor $v$. 
		\li Let $S_{U',V'}\subseteq V'\cap \Gamma(U')$ contain the $\rho\card{V'}$ vertices with largest $f_v$.
	\li \Return $bias^{V'}_{U'}\defeq (1/\rho\card{V'})\sum_{v\in S_{U',V'}}{f_v}$ heads.
\end{algbox}

In the next lemma we prove that $bias_{U'}^{V'}$ is likely to approximate $bias_{U'}$ well. For future use we phrase a more general statement than we need here, addressing $U'$ that defines a slightly faulty coin as well.

\begin{lemma}\label{l:bias-V'}
Assume that $\card{\Gamma(U')}\geq (1-\gamma')\rho\card{V}$, where $\gamma' = \Theta(\gamma)$ and $\gamma,\gamma'\leq 1/4$. For a uniform $V'\subseteq V$, except with probability exponentially small in $\rho\gamma^2\card{V'}$, 
$$\card{bias_{U'}^{V'} - bias_{U'}}\leq 3\gamma + 2\gamma'.$$
\end{lemma}
\begin{proof} 
By a multiplicative Chernoff bound, except with probability exponentially small in $\rho\gamma^2 \card{V'}$, there are $(1\pm \gamma \pm \gamma')\rho\card{V'}$ vertices in $V'\cap S_{U'}$. Let us focus on this event.

By a Hoeffding bound, except with probability exponentially small in $\rho\gamma^2 \card{V'}$, we have $$\card{\frac{1}{\card{V'\cap S_{U'}}}\sum_{v\in V'\cap S_{U'}} f_v - \frac{1}{\card{S_{U'}}}\sum_{v\in S_{U'}} f_v} \leq \gamma.$$
Hence,
$$\card{\frac{1}{\rho\card{V'}}\sum_{v\in V'\cap S_{U'}} f_v - \frac{1}{\card{S_{U'}}}\sum_{v\in S_{U'}} f_v} \leq 3\gamma+2\gamma'.$$
The lemma follows.
\end{proof}
%In the sequel we'll use the notation $S_{U',V'}$ and $bias^{V'}_{U'}$ for any %subset $V'\subseteq V$.

The algorithm for finding an approximate clique using \textsc{Clique-Coin-Toss} is described in Figure~\ref{alg:amplified-approx-clique}.
Note that the coin tossing algorithm satisfies the conditions of simulated tossing (Definition~\ref{d:simulate-toss}).
Lemma~\ref{l:U'-approx} and Lemma~\ref{l:C-degree} ensure that with constant probability over the choice of $U$, for $U'$ as specified in Section~\ref{s:dense-randomized}, we have $bias_{U'}\geq 1-2\varepsilon$ for one of the $U'\subseteq U$. 
Moreover, a coin with bias at least $1-c\varepsilon$ yields a set which is at least $1-c\varepsilon/\rho$-dense, and this set can be computed in $O(\card{V}^2)$ time. 
Therefore, the algorithm in Figure~\ref{alg:amplified-approx-clique} gives an algorithm for finding an approximate clique that errs with probability exponentially small in $\card{V}$ and runs in time $\tilde{O}(\card{V}^2 2^{O(1/\varepsilon^2\rho)})$.  
This proves part of Theorem~\ref{t:clique} repeated below for convenience (note that $\varepsilon$ in Theorem~\ref{t:clique} is replaced with $O(\varepsilon/\rho)$ here).

\begin{theorem} There is a Las Vegas algorithm that given a graph $G=(V,E)$ with a clique on $\rho\card{V}$ vertices and given $0<\rho,\varepsilon<1$, finds a set of $\rho\card{V}$ vertices and density $1-O(\varepsilon/\rho)$, except with probability exponentially small in $\card{V}$. The algorithm runs in time $\tilde{O}(\card{V}^2 2^{O(1/(\varepsilon^2\rho))})$. 
\end{theorem}

The remainder of the section constructs an oblivious verifier for \textsc{Find-Approximate-Clique} and uses it to prove the second part of Theorem~\ref{t:clique} (a deterministic algorithm).
First we describe the sketch and its properties, then we devise an oblivious verifier for \textsc{Clique-Coin-Toss}, and finally we describe the verifier for \textsc{Find-Approximate-Clique}.

\begin{algbox}{6.5in}{amplified-approx-clique}{An algorithm for finding an approximate clique in a graph $G=(V,E)$ that contains a clique on $\rho\card{V}$ vertices. The error probability of the algorithm is exponentially small in $\card{V}$.}
	\Procname{\textsc{Find-Approximate-Clique}$(G=(V,E),\rho,\varepsilon)$}
	\li Set $u= \ceil{100/(\varepsilon^2\rho)}$.
	\li Set $i_f= \log((\card{V}+u)/\varepsilon^2)+\Theta(\log\log((\card{V}+u)/\varepsilon))$;\;\; $\beta = \varepsilon/i_f$.
	\li Set $i_0= \log(u/\beta^2)+\Theta(1)$.
	\li Sample $U\subseteq V$, $\card{U} = u$. 
	\li\label{s:loop-begin} \For $i=i_0,i_0+1,\ldots, i_f$ \Do 
		\li Set $k = 2^i$.
		\li \For all $U'\subseteq U$, $\card{U'}\geq(\rho/2)\card{U}$ \Do
					\li \textsc{Clique-Coin-Toss}$(G,U',\rho,\beta^2 k,\gamma = \beta/100)$. If fails, skip this $U'$.
		\End
		\li\label{s:loop-end} If the fraction of heads is less than $1- 2\varepsilon - i\beta$ for all (non-skipped) $U'$, restart.
	\End
	\li \Return set $S_{U'}$ of density at least $1-3\varepsilon/\rho$ if such exists.
\end{algbox}

\subsection{A Sketch for Approximate Clique}

The sketch for a given $G$ contains, for some carefully chosen set $R$ of $\poly(\log\card{V},1/\varepsilon,1/\rho)$ vertices, the bipartite graph $G_R=(R,V,E_R)$ that contains all the edges of $G$ that at least one of their endpoints falls in $R$. 
The set $R$ is chosen so it allows the verifier to estimate the $f_v$'s corresponding to different sets $U'\subseteq V$. Note that the size of the sketch is $\card{R}\card{V}$. 

Let $U'\subseteq U$.
For every $v\in V$ we denote by $f_v$ the fraction of vertices in $\Gamma(U')$ that neighbor $v$. For $V'\subseteq V$, let $S_{U',V'}\subseteq V'\cap \Gamma(U')$ denote the $\rho\card{V'}$ elements $v\in V'\cap \Gamma(U')$ with largest $f_v$ (pad with dummy vertices with $0$ neighbors if needed).
Let $bias_{U'}^{V'} \defeq (1/\rho\card{V'})\sum_{v\in S_{U',V'}}f_v$.
For $v\in V$ let $\tilde{f}_v$ denote the fraction of $\Gamma(U')\cap R$ vertices that neighbor $v$ among all vertices in $\Gamma(U') \cap R$.
For $V'\subseteq V$, let $\tilde{S}_{U',V'}$ be the $\rho\card{V'}$ vertices $v\in V'$ with largest $\tilde{f}_v$ (pad with dummy vertices with $0$ neighbors if needed). Let $\tilde{bias}_{U'}^{V'} \defeq (1/\rho\card{V'})\sum_{v\in \tilde{S}_{U',V'}}\tilde{f}_v$.

In the lemma we use $u,\rho,\gamma$ from \textsc{Find-Approximate-Clique} in Figure~\ref{alg:amplified-approx-clique}.

\begin{lemma}[Sketch]\label{l:clique-sketch}
There exists $R\subseteq V$, $\card{R} = O(u\log \card{V}/\rho\gamma^2)$, such that for every $U'\subseteq V$,  $\card{U'}\leq u$,
\begin{enumerate}
\item 
If 
$\card{\Gamma(U')}\geq \rho\card{V}$, then $\card{R\cap \Gamma(U')}\geq (1-\gamma)\rho\card{R}$, whereas
if
$\card{\Gamma(U')}< (1-2\gamma)\rho\card{V}$, then $\card{R\cap \Gamma(U')}< (1-\gamma)\rho\card{R}$.

\item Suppose that $\card{\Gamma(U')}\geq (1-2\gamma)\rho\card{V}$. Then, 
for every $v\in V$, we have $\card{\tilde{f}_v - f_v}\leq \gamma$.

\item Suppose that $\card{\Gamma(U')}\geq (1-2\gamma)\rho\card{V}$. Then, $\card{bias_{U'}^{R}-bias_{U'}}\leq 7\gamma$.
\end{enumerate}
\end{lemma}
\begin{proof}
Pick $R\subseteq V$ uniformly at random. Let $U'\subseteq U$, $\card{U'}\leq u$. By a multiplicative Chernoff bound,
if 
$\card{\Gamma(U')}\geq \rho\card{V}$, then $\card{R\cap \Gamma(U')}\geq (1-\gamma)\rho\card{R}$, except with probability exponentially small in $\rho\gamma^2\card{R}$.
If $\card{\Gamma(U')}< (1-2\gamma)\rho\card{V}$, then $\card{R\cap \Gamma(U')}\leq (1-\gamma)\rho\card{R}$ except with probability exponentially small in $\rho\gamma^2\card{R}$. 

Suppose that $\card{\Gamma(U')}\geq (1-2\gamma)\rho\card{V}$.
Let $v\in V$. By a multiplicative Chernoff bound, except with probability exponentially small in $\rho\gamma^2\card{R}$, we have $\card{\Gamma(U')\cap R}\geq (1-3\gamma)\card{R}$.
By a Chernoff bound, except with probability exponentially small in $\rho\gamma^2\card{R}$, we have $$\card{\tilde{f}_v - f_v}\leq \gamma.$$ 
By a union bound over all $v$ and by the choice of $\card{R}$, the last inequality holds for all $v\in V$ except with probability 
exponentially small in $\rho\gamma^2\card{R}$.

By Lemma~\ref{l:bias-V'}, if $\card{\Gamma(U')}\geq (1-2\gamma)\rho\card{V}$, except with probability exponentially small in $\rho\gamma^2\card{R}$, we have $$\card{bias_{U'}^{R} - bias_{U'}} \leq 7\gamma.$$ 

Since there are less than $\card{V}^{u}$ choices for $U'$, it follows from a union bound that there exists $R$ for which all three items hold for all $U'\subseteq V$, $\card{U'}\leq u$.
\end{proof}

\begin{lemma}\label{l:density-approx} 
Suppose that $\card{\Gamma(U')}\geq (1-2\gamma)\rho\card{V}$.
For any $V'\subseteq V$,
$$\card{\tilde{bias}_{U'}^{V'} - bias_{U'}^{V'}}\leq 2\gamma.$$   
\end{lemma}
\begin{proof}
By Lemma~\ref{l:clique-sketch},
the contribution from $v\in V'$ in $\tilde{S}_{U',V'}\cap S_{U',V'}$ is at most $\gamma$ since $\card{\tilde{f}_v - f_v}\leq\gamma$. It remains to bound the contribution from other elements $v\in V'$ that are either in $\tilde{S}_{U',V'} - S_{U',V'}$ or in $S_{U',V'} - \tilde{S}_{U',V'}$. Pair those vertices arbitrarily, and consider a single pair $v_2 \in \tilde{S}_{U',V'} - S_{U',V'}$ and $v_1\in S_{U',V'} - \tilde{S}_{U',V'}$. We know that $f_{v_1} \geq f_{v_2}\geq \tilde{f}_{v_2} - \gamma$, so $\tilde{f}_{v_2} - f_{v_1}\leq\gamma$. Similarly, $\tilde{f}_{v_2}\geq \tilde{f}_{v_1}\geq f_{v_1} - \gamma$, so $f_{v_1} -\tilde{f}_{v_2} \leq \gamma$. In any case, $\card{f_{v_1} - \tilde{f}_{v_2}}\leq \gamma$. The triangle inequality implies the lemma.
\end{proof}

\begin{corollary}\label{c:bias-sketch}
For every $U'\subseteq V$, $\card{U'}\leq u$, either $\card{R\cap \Gamma(U')}<(1-2\gamma)\rho\card{V}$, or
$$\card{\tilde{bias}_{U'}^{R} - bias_{U'}}\leq 9\gamma.$$   
\end{corollary}
\begin{proof} If $\card{R\cap \Gamma(U')}\geq(1-2\gamma)\rho\card{V}$, by Lemma~\ref{l:clique-sketch}, we have 
$\card{bias_{U'}^{R} - bias_{U'}}\leq 7\gamma$. 
By Lemma~\ref{l:density-approx}, $\card{\tilde{bias}_{U'}^{R} - bias_{U'}^{R}}\leq 2\gamma$. 
The claim follows.
\end{proof}

Interestingly, our construction of the sketch is randomized, yet it will allow us to obtain a deterministic algorithm. The reason is that we only need the existence of a sketch describing an input so we can take a union bound over all possible sketches.

\subsection{Obliviously Checking $V'$}

Next we show how we can check the sample $V'$ of \textsc{Clique-Coin-Toss} using the sketch. 
The oblivious verifier receives a sketch $G_R$ of the graph $G$, the rest of the input of \textsc{Clique-Coin-Toss} and the randomness $V'$ used by the algorithm. 
The verifier accepts iff $bias_{U'}^{V'}$ is approximately $bias_{U'}$. It uses the sketch to approximate $bias_{U'}$ via $\tilde{bias}_{U'}^R$. 
It is described in Figure~\ref{alg:verifier-coin-toss}.

\begin{algbox}{6.5in}{verifier-coin-toss}{
An oblivious verifier for \textsc{Clique-Coin-Toss}.
}
\Procname{\textsc{Oblivious-Verifier-Clique-Coin-Toss}$(G_R,U',\rho,k,\gamma,V')$}
		\li If $\card{R\cap \Gamma(U')}<(1-\gamma)\rho\card{R}$  \Then
		\li\label{s:clique-toss-verify-reject}	Fail.
		\End
%		\li Let $V'\subseteq V$ be the sample of \textsc{Clique-Coin-Toss}.
		\li For all $v\in V'$ compute the fraction $\tilde{f}_v$ of vertices in $\Gamma(U')\cap R$ that neighbor $v$.
		\li Let $\tilde{S}_{U',V'}$ be the $\rho\card{V'}$ vertices $v\in V'$ with largest $\tilde{f}_v$.
		\li Let $\tilde{bias}_{U'}^{V'} \defeq (1/\rho\card{V'})\sum_{v\in \tilde{S}_{U',V'}}\tilde{f}_v$.
	\li Accept iff $\card{\tilde{bias}_{U'}^{V'}-\tilde{bias}_{U'}^R}\leq 18\gamma$.
\end{algbox}
Recall that by Lemma~\ref{l:clique-sketch}, if the coin is non-faulty and 
$\card{\Gamma(U')}\geq \rho\card{V}$, then $\card{R\cap \Gamma(U')}\geq (1-\gamma)\rho\card{R}$, so \textsc{Oblivious-Verifier-Clique-Coin-Toss} does not fail in Step~\ref{s:clique-toss-verify-reject}.
Moreover, if
$\card{\Gamma(U')}< (1-2\gamma)\rho\card{V}$, then $\card{R\cap \Gamma(U')}< (1-\gamma)\rho\card{R}$, and \textsc{Oblivious-Verifier-Clique-Coin-Toss} necessarily fails.

\begin{lemma}\label{l:OVCCT} The following hold:
\begin{enumerate}
\item If \textsc{Oblivious-Verifier-Clique-Coin-Toss} accepts then $\card{\Gamma(U')} \geq (1-2\gamma)\rho\card{V}$ and $V'$ sampled by \textsc{Clique-Coin-Toss} satisfies $\card{bias_{U'}^{V'}-bias_{U'}}\leq 29\gamma$.
 
\item If $\card{R\cap \Gamma(U')}\geq (1-\gamma)\rho\card{V}$, then the probability that \textsc{Oblivious-Verifier-Clique-Coin-Toss} rejects is exponentially small in $k$.
\end{enumerate}
\end{lemma}
\begin{proof}
Toward the second item, suppose that $\card{R\cap \Gamma(U')}\geq (1-\gamma)\rho\card{V}$. By Lemma~\ref{l:density-approx}, we have $\card{\tilde{bias}_{U'}^{V'} - bias_{U'}^{V'}}\leq 2\gamma$. 
By Corollary~\ref{c:bias-sketch}, we have $\card{\tilde{bias}_{U'}^{R} - bias_{U'}}\leq 9\gamma$. By Lemma~\ref{l:bias-V'}, we have $\card{bias_{U'}^{V'} - bias_{U'}}\leq 7\gamma$, 
except with probability exponentially small in $k$.
The low probability of rejection follows. 

Toward the first item, suppose that \textsc{Oblivious-Verifier-Clique-Coin-Toss} accepts, so $\card{R\cap \Gamma(U')}\geq (1-\gamma)\rho\card{R}$ and $\card{\tilde{bias}_{U'}^{V'}-\tilde{bias}_{U'}^R}\leq 18\gamma$. By Lemma~\ref{l:density-approx}, we have $\card{\tilde{bias}_{U'}^{V'} - bias_{U'}^{V'}}\leq 2\gamma$. 
By Corollary~\ref{c:bias-sketch}, we have $\card{\tilde{bias}_{U'}^{R} - bias_{U'}}\leq 9\gamma$.
Thus, we have $\card{bias_{U'}^{V'} - bias_{U'}}\leq 29\gamma$.

\end{proof}

\subsection{An Oblivious Verifier for Approximate Clique}\label{s:clique-verifier}

The verifier is unable to follow the execution of \textsc{Find-Approximate-Clique} nor compute its output, since it can't tell exactly how many heads \textsc{Clique-Coin-Toss} yields. The verifier can be probably approximately correct about the fraction of heads, but it is likely that during the execution of \textsc{Find-Approximate-Clique} some of its predictions would be false, thereby changing the course of execution.
It may seem that under these conditions the verifier cannot check the randomness of the algorithm, but this is not so. The key idea is that the verifier is not limited computationally and can try all possible executions of the algorithm (i.e., the outcomes of all possible restart decisions). 

\begin{rem} The algorithm \textsc{Find-Approximate-Clique} uses its randomness as a stream of random bits, and uses independent randomness between restarts. The oblivious verifier for a single execution simulates a possible run of the algorithm and follows the algorithm in its use of the randomness. Different executions use the same randomness. 
\end{rem}

The verifier maintains a set $\mathcal{G}$ of possible input graphs $G$ that are consistent with the execution up to this step (initially $\mathcal{G}$ contains all the input graphs that are consistent with the sketch). If the set of inputs becomes empty, then the execution is designated {\em infeasible}. Otherwise the execution is designated {\em feasible}.
Additionally, the verifier maintains $counter$ such that the probability of the execution is at most exponentially small in $counter$ (initially, $counter=0$). If $counter$ becomes too large, the verifier rejects the execution. If none of the feasible executions get rejected, the verifier accepts. The verifier for a single execution (a single fixing of guesses) is described in Figure~\ref{alg:verifier-approx-clique}. Note that we use the shorthand \textsc{OVCCT} for \textsc{Oblivious-Verifier-Clique-Coin-Toss} and that the verifier uses the parameters $i_0,i_f,\beta$ of the algorithm. The final verifier is described in Figure~\ref{alg:verifier-final-clique}.

\begin{algbox}{6.5in}{verifier-approx-clique}{An oblivious verifier for a single execution of \textsc{Find-Approximate-Clique} (an execution is defined by the outcomes of guesses).
We use the shorthand \textsc{OVCCT} for \textsc{Oblivious-Verifier-Clique-Coin-Toss}.}
	\Procname{\textsc{Oblivious-Verifier-Clique-Execution}$(\mathcal{G},G_R=(R,V,E_R),\rho,\varepsilon,r,counter)$} 
	\li \If $\mathcal{G}=\emptyset$ \Then
			\li \Return infeasible.
	\End
	\li \If $counter > \card{V}$ \Then
			\li \Return reject.
	\End
	\li Extract from $r$ the sample $U\subseteq V$ of the algorithm.
	\li \If $\tilde{bias}_{U'}^R <1-2\varepsilon - \beta/2 + 9\gamma$ for all $U'$ such that $\card{R\cap \Gamma(U')} \geq (1-\gamma)\rho\card{V}$ \Then
		\li\label{s:+1}$counter \leftarrow counter+ 25/\varepsilon$.
	\End
	\li\For $i=i_0,i_0+1,\ldots,i_f$ \Do
		\li Set $k = 2^i$.
		\li \For all $U'\subseteq U$, $\card{U'}\geq (\rho/2)\card{U}$ \Do
				\li Extract from $r$ the randomness $V'$ for \textsc{Clique-Coin-Toss}.
				\li \textsc{OVCCT}$(G_R,U',\rho,\beta^2 k,\gamma=\beta/100,V')$
		\End
		\li \If $\exists U', \card{R\cap \Gamma(U')}\geq (1-\gamma)\rho\card{V}$ such that \textsc{OVCCT} rejects \Then
			\li\label{s:+k}$counter \leftarrow counter+\beta^2 k - u$. 
				\End
		\li\label{s:G-update}Guess if $\max\sett{bias_{U'}^{V'}}{\card{\Gamma(U')}\geq\rho\card{V}}< 1-2\varepsilon - i\beta$ and update $\mathcal{G}$ accordingly.
		\li \If guessed true \Then
				\li Restart maintaining $\mathcal{G}$ and $counter$.
		\End
	\End
	\li\label{clique:check}\Return accept iff $\exists U'\subseteq U$, $\card{R\cap \Gamma(U')}\geq (1-\gamma)\rho\card{V}$, such that $\tilde{bias}^R_{U'}\geq 1 - 3\varepsilon - 9\gamma$.
\end{algbox}

\begin{algbox}{6.5in}{verifier-final-clique}{The final oblivious verifier for approximate clique.}
	\Procname{\textsc{Oblivious-Verifier-Approximate-Clique}$(G_R=(R,V,E_R),\rho,\varepsilon,r)$} 
	\li Let $\mathcal{G}$ contain all the graphs that are consistent with $G_R$.
	\li Try all guesses in \textsc{Oblivious-Verifier-Clique-Execution}$(\mathcal{G},G_R,\rho,\varepsilon,r,0)$.
	\li Accept iff all feasible executions accept.
\end{algbox}

Next we analyze \textsc{Oblivious-Verifier-Approximate-Clique}.

\begin{lemma}
If \textsc{Oblivious-Verifier-Approximate-Clique} accepts on a sketch of a graph $G$ and randomness $r$, then \textsc{Find-Approximate-Clique}\begin{footnote}{Note that we argue about a version of \textsc{Find-Approximate-Clique} that checks a relaxed condition on the density such that the condition is satisfied by a coin of bias $1- 3\varepsilon-18\gamma$ (as opposed to the version of Figure~\ref{alg:amplified-approx-clique}).}\end{footnote} necessarily finds $U'\subseteq V$ with $bias_{U'}\geq 1- 3\varepsilon-18\gamma$ when invoked on $G$ and $r$ with the same parameters $\rho$ and $\varepsilon$.
\end{lemma}
\begin{proof}
If \textsc{Oblivious-Verifier-Approximate-Clique} accepts with $G$'s sketch and randomness $r$, then, in particular, the execution of \textsc{Oblivious-Verifier-Clique-Execution} with the guesses that correspond to the run of \textsc{Find-Approximate-Clique} on $G$ and $r$ results in $U$ such that for some $U'\subseteq U$, $\card{R\cap \Gamma(U')}\geq (1-\gamma)\rho\card{V}$, it holds that $\tilde{bias}_{U'}^R\geq 1 -3\varepsilon - 9\gamma$. By Corollary~\ref{c:bias-sketch}, $bias_{U'}\geq 1- 3\varepsilon-18\gamma$.
\end{proof}

Next we argue that \textsc{Oblivious-Verifier-Approximate-Clique} rejects with probability exponentially small in $\card{V}$.

\begin{lemma}\label{l:clique-soundness} The following hold:
\begin{itemize}
\item For any guesses, the probability that \textsc{Oblivious-Verifier-Clique-Execution} rejects is exponentially small in $\card{V}$.

\item \textsc{Oblivious-Verifier-Clique-Execution} makes a number of guesses that is sufficiently smaller than $\card{V}$ (the ratio between the number of guesses and $\card{V}$ can be made arbitrarily small by lowering $\varepsilon$ by a constant factor and by increasing $i_0$ by a constant).
\end{itemize}
\end{lemma}

The correctness of the final oblivious verifier follows from a union bound over all possible choices of guesses. If the number of guesses is sufficiently smaller than $\card{V}$, the probability that \textsc{Oblivious-Verifier-Approximate-Clique} rejects is exponentially small in $\card{V}$. 

In order to show the two items above we show that three invariants hold. The invariants are in Lemmas~\ref{i:clique-1},~\ref{i:clique-2} and~\ref{i:clique-3}. 
\begin{lemma}\label{i:clique-1}
Throughout the execution of \textsc{Oblivious-Verifier-Clique-Execution}: $\mathcal{G}$ contains all the graphs that are consistent with the sketch and guesses so far.
\end{lemma}
\begin{proof}
The invariant holds since $\mathcal{G}$ is initialized to contain all graphs consistent with the sketch, and since after each guess $\mathcal{G}$ is updated to contain only those graphs in $\mathcal{G}$ that are consistent with the guess. 
\end{proof}

\begin{lemma}\label{i:clique-2} 
Throughout the execution of \textsc{Oblivious-Verifier-Clique-Execution}:
The probability of reaching $counter = c$ is exponentially small in $c$.  
\end{lemma}
\begin{proof}
$counter$ is initially $0$. The increase in step~\ref{s:+1} is justified as follows. When this step occurs, $\tilde{bias}_{U'}^R <1-2\varepsilon - \beta/2 + 9\gamma$ for all $U'$ such that $\card{R\cap \Gamma(U')} \geq (1-\gamma)\rho\card{V}$.
By 
Corollary~\ref{c:bias-sketch}, for all $U'$ such that $\card{\Gamma(U')}\geq(1-2\gamma)\rho\card{V}$, we have $\card{\tilde{bias}_{U'}^R - bias_{U}}\leq 9\gamma$. Hence, when step~\ref{s:+1} occurs for all $U'$ with $\card{\Gamma(U')}\geq \rho\card{V}$, it holds that 
$bias_{U'}< 1-2\varepsilon$ (here we also use Lemma~\ref{l:clique-sketch}).
By Lemma~\ref{l:clique-const-err}, for every graph in $\mathcal{G}$, the probability that this happens is $e^{-25/\varepsilon}$, and therefore the increase in $counter$ is justified.
The increase in Step~\ref{s:+k} follows from the correctness of \textsc{Oblivious-Verifier-Clique-Coin-Toss} (Note that we take a union bound over $2^u$ possible $U'$).
\end{proof}

\begin{lemma}\label{i:clique-3} If \textsc{Oblivious-Verifier-Clique-Execution} restarts in phase $i$ then $counter$ increases by at least $\beta^2 2^{i-1} -u$ either in the restart phase or in the previous phase, or it is the first phase and $counter$ increases by $25/\varepsilon$. 
\end{lemma}
\begin{proof}
Suppose that there is a restart in phase $i$. Let $k=2^{i}$.  
By Lemma~\ref{l:simulated-coin}, either in this iteration, or in the previous iteration, there exists a non-faulty coin in the group that yielded a fraction of heads that deviates from its bias by an additive $\beta/2$, or it's the first phase and a group with bias at most $1-2\varepsilon - \beta/2$ was picked. Next we handle each of these cases. 

\begin{enumerate}
\item Suppose that $i=i_0$ and $\max\sett{bias_{U'}}{\card{\Gamma(U')}\geq\rho\card{V}}< 1-2\varepsilon - \beta/2$. 
By Lemma~\ref{l:clique-sketch}, $\card{R\cap \Gamma(U')}\geq (1-\gamma)\rho\card{R}$, and, hence, by Corollary~\ref{c:bias-sketch}, $\card{\tilde{bias}_{U'}^R - bias_{U'}}\leq 9\gamma$. This implies that $\tilde{bias}_{U'}^R < 1-2\varepsilon - \beta/2 + 9\gamma$ and that $counter$ increases in step~\ref{s:+1}.

\item Suppose that there exists $U'\subseteq U$ with $\card{\Gamma(U')}\geq\rho\card{V}$ such that $bias_{U'}^{V'}$ deviates from $bias_{U'}$ by at least $\beta/2$ in phase $i$ or $i-1$. By Lemma~\ref{l:OVCCT} (and Lemma~\ref{l:clique-sketch}), since $29\gamma< \beta/2$ and \textsc{Clique-Coin-Toss} does not fail on $U'$, we know that \textsc{Clique-Coin-Toss} rejects and that $counter$ increases in Step~\ref{s:+k}.

\end{enumerate}
The lemma follows.
\end{proof}

Next we prove Lemma~\ref{l:clique-soundness} from the invariants in Lemmas~\ref{i:clique-1},~\ref{i:clique-2} and \ref{i:clique-3}.

\begin{proof}(of Lemma~\ref{l:clique-soundness} from invariants)
Rejection is caused either by $counter$ reaching $\card{V}$, which happens with probability exponentially small in $\card{V}$ by Lemma~\ref{i:clique-2}, or by \textsc{Find-Approximate-Clique}, running on any of the graphs in $\mathcal{G}$ (recall that $\mathcal{G}\neq\emptyset$), returning, by Corollary~\ref{c:bias-sketch}, $U'$ such that $bias_{U'}< 1-3\varepsilon$. 
We already saw that the latter has probability exponentially small in $\card{V}$ for any graph in $\mathcal{G}$ in Section~\ref{s:clique-amplified}. 

The second item follows from Lemma~\ref{i:clique-3} and since $counter$ increases by large increments and $counter\leq\card{V}$. The increments are either at least $\beta^2 2^{i_0 - 1 } - u$ (which corresponds to the constant term in $i_0$) or $25/\varepsilon$.
\end{proof}

A non-uniform deterministic algorithm for approximate clique follows, concluding the proof of Theorem~\ref{t:clique} (note that $\varepsilon$ in Theorem~\ref{t:clique} is replaced with $O(\varepsilon/\rho)$ here). 
\begin{theorem}\label{t:clique-det}
There is a deterministic non-uniform algorithm that given $0<\rho,\varepsilon<1$ and a graph $G=(V,E)$ with a clique on $\rho\card{V}$ vertices, finds a set of $\rho\card{V}$ vertices and density $1-O(\varepsilon/\rho)$. The algorithm runs in time $\tilde{O}(\card{V}^2 2^{O(1/(\varepsilon^2\rho))})$.
\end{theorem}

\begin{rem}\label{r:clique-alternative}
There is a an alternative way to prove Theorem~\ref{t:clique-det} based on the biased coin algorithm. It involves a more complicated randomized algorithm that can achieve lower error probability and no sketching. First we devise a coin tossing algorithm that achieves error probability exponentially small in $k$ for any $1\leq k\leq\card{V}$ while running in time $O(k\card{U'}\poly(1/\rho,1/\gamma))$. This algorithm picks $V'\subseteq V$ like \textsc{Clique-Coin-Toss}, but then instead of computing $bias_{U'}^{V'}$ directly, it estimates $bias_{U'}^{V'}$ by picking for every $v\in V'$ a small independent sample to estimate the fraction of $\Gamma(U')$ vertices that neighbor it. The observation is that it suffices that the estimates for most $v\in V'$ are accurate in order for the estimate for $bias_{U'}^{V'}$ to be accurate. By repeating this coin tossing algorithm $r$ times, we can obtain a coin tossing algorithm that runs in time $O(kr\card{U'}\poly(1/\rho,1/\gamma))$ and achieves error probability exponentially small in $kr$ for any $r\geq 1$. Via the biased coin algorithm, we obtain a Las Vegas algorithm that runs in time $\tilde{O}(\card{V}^2 2^{O(1/(\varepsilon^2\rho))})$ and achieves error probability exponentially small in $\card{V}^2$. This implies a deterministic non-uniform algorithm that runs in the same time.
\end{rem}

\section{Free Games}\label{s:free}

\subsection{A Simple Randomized Algorithm}

First we describe a simple randomized algorithm with constant error probability for free games based on the sampling idea in the \textsc{Max-Cut} algorithm. A similar algorithm appeared in~\cite{AIM}.
The main idea of the algorithm is as follows. 
We sample a small $S\subseteq X$ and enumerate over all possible labelings $h:S\to\Sigma$. Each labeling induces a labeling to all vertices as follows.
\begin{definition}[Induced labeling]
Let $\mathcal{G}$ be a free game on a graph $G=(X,Y,X\times Y)$, alphabet $\Sigma$ and constraints $\set{\pi_e}$. Let $S\subseteq X$ and let $h:S\to\Sigma$ be a labeling to $S$. 
\begin{itemize}
\item The induced labeling $f_{S,h,Y}:Y\to\Sigma$ is defined as follows: For every $y\in Y$ let $f_{S,h,Y}(y)$ be the label $\sigma\in\Sigma$ that maximizes the fraction of edges $e=(s,y)\in S\times \set{y}$ such that $(h(s),\sigma)\in\pi_e$ (ties are broken arbitrarily). 
\item The induced labeling $f_{S,h,X}:X\to\Sigma$ is defined as follows: For every $x\in X$ let $f_{S,h,X}(x)$ be the label $\sigma\in\Sigma$ that maximizes the fraction of edges $e=(x,y)\in \set{x}\times Y$ such that $(\sigma,f_{S,h,Y}(y))\in\pi_e$ (ties are broken arbitrarily).
\end{itemize}
\end{definition}
Note that for each $y\in Y$ computing $f_{S,h,Y}(y)$ takes time $O(\card{\Sigma}\card{S})$. For each $x\in X$ computing $f_{S,h,X}(x)$ takes time $O(\card{\Sigma}\card{Y})$.

We will argue below that if $h$ is the restriction of an optimal assignment to $\mathcal{G}$, then the induced labeling is likely to approximately achieve the value of $\mathcal{G}$. 

\begin{lemma}[Sampling]\label{l:Russell}
Let $\mathcal{G}$ be a free game on a graph $G=(X,Y,X\times Y)$ and with alphabet $\Sigma$. Let $\varepsilon,\delta>0$. Then for a uniform $S\subseteq X$, $\card{S}=\ceil{\log(\card{\Sigma}/\varepsilon\delta)/\varepsilon^2}$, with probability at least $1-\delta$, there exists $h:S\to\Sigma$ such that $f_{S,h,X}$, $f_{S,h,Y}$ satisfy at least $val(\mathcal{G}) - 2\varepsilon$ fraction of the edges in $G$. 
\end{lemma}
\begin{proof}
There is a labeling $f_X^*:X\to\Sigma$, $f_Y^*:Y\to\Sigma$ that achieves the value of $\mathcal{G}$, namely, $val_{f_X^*,f_Y^*}(\mathcal{G}) = val(\mathcal{G})$.
Let $h:S\to\Sigma$ be the restriction of $f^*_X$ to $S$. Let $y\in Y$.
Let $\sigma\in\Sigma$. By a Chernoff bound except with probability $\varepsilon\delta/\card{\Sigma}$, the fraction of edges $e=(s,y)\in S\times \set{y}$ such that $(h(s),\sigma)\in \pi_e$, is the same up to an additive $\varepsilon$ as the fraction of edges $e=(x,y)\in X\times \set{y}$ such that $(f_X^*(x),\sigma)\in\pi_e$. By a union bound over all $\sigma\in\Sigma$, for each $y\in Y$, except with probability at most $\varepsilon\delta$ over $S$, this holds for all $\sigma\in\Sigma$. 
In other words, for a uniform $S$, the expected fraction of $y\in Y$ for which this holds is at most $\varepsilon\delta$. Thus, with probability at most $\delta$ over the choice of $S$, for at least $1-\varepsilon$ fraction of the $y\in Y$, this holds. Therefore, except for at most $\delta$ fraction of the $S$, we have
$val_{f^*_{X},f_{S,h,Y}}(\mathcal{G})\geq val_{f_X^*,f_Y^*}(\mathcal{G}) - 2\varepsilon = val(\mathcal{G}) - 2\varepsilon$. 
The lemma follows noticing that $val_{f_{S,h,X},f_{S,h,Y}}(\mathcal{G})\geq val_{f^*_{X},f_{S,h,Y}}(\mathcal{G})$.
\end{proof}

\subsection{A Randomized Algorithm With Exponentially Small Error Probability}

We think of sampling $S\subseteq X$ as picking a group of coins. The group has a coin per $h:S\to\Sigma$. The bias of the coin is the fraction of edges satisfied by $f_{S,h,X}$ and $f_{S,h,Y}$. 
One tosses a coin $k$ times by picking roughly $k$ vertices $X'\subseteq X$ and estimating the success of $f_{S,h,X}$ and $f_{S,h,Y}$ on edges that touch $X'$. The coin tossing algorithm is described in Figure~\ref{alg:toss-free-coin}. For a deviation parameter $\gamma$ dictating by how much the coin toss deviates from the actual bias and for $k$ dictating the error probability, the toss runs in time $k\card{Y}\cdot\poly(\card{\Sigma},1/\varepsilon,1/\gamma)$.

\begin{algbox}{6.5in}{toss-free-coin}{A coin toss picks vertices at random and estimates how many of the edges that touch the sample are satisfied by $f_{S,h,X}$ and $f_{S,h,Y}$.}
	\Procname{\textsc{Free-Toss-Coin}$(\mathcal{G},S,h,k,\gamma)$}
	\li Compute $f_{S,h,Y}(y)$ for all $y\in Y$.
	\li Pick $X'\subseteq X$, $\card{X'}= \ceil{(k+\card{S}\log\card{\Sigma})/\gamma^2}$, uniformly at random.
	\li For all $x\in X'$ and $\sigma\in\Sigma$ compute the fraction of $y\in Y$ so $(\sigma,f_{S,h,Y}(y))\in\pi_{(x,y)}$.
	\li For all $x\in X'$ let $g_{S,h,x}$ be the max over $\sigma\in\Sigma$ of the fractions.
	\li \Return $bias_{S,h}^{X'}\defeq (1/\card{X'})\sum_{x\in X'} g_{S,h,x}$ fraction heads.
\end{algbox}

Let $bias_{S,h}$ be the fraction of edges satisfied by $f_{S,h,X}$ and $f_{S,h,Y}$. 
For $X'\subseteq X$, let $bias_{S,h}^{X'}$ be the fraction of edges that touch $X'$ and are satisfied by $f_{S,h,X}$, $f_{S,h,Y}$.
For each $x\in X$ let $g_{S,h,x}$ be the fraction of edges that touch $x$ and are satisfied by $f_{S,h,X}$, $f_{S,h,Y}$. 
We have $bias_{S,h} = (1/\card{X})\sum_{x\in X}{g_{S,h,x}}$, and 
and $(1/\card{X'})\sum_{x\in X'}{g_{S,h,x}} = bias_{S,h}^{X'}$. 
The following lemma shows that the estimate $bias_{S,h}^{X'}$ of the coin tossing algorithm typically doesn't deviate much from the actual bias $bias_{S,h}$. 
\begin{lemma}\label{l:free-sample}%The following hold.
%\begin{itemize}
%\item 
Except with probability exponentially small in $k$, for all $h:S\to\Sigma$,
$$\card{bias_{S,h}^{X'} - bias_{S,h}}\leq \gamma.$$
%\item Except with probability $\card{\Sigma}^{\card{S}}(e\delta/\gamma)^{\gamma %\card{X'}}$, for all $h:S\to\Sigma$,
%$$\card{\tilde{bias}_{S,h}^{X'} - bias_{S,h}^{X'}}\leq 2\gamma.$$
%\end{itemize}
\end{lemma}
\begin{proof} 
By a Hoeffding bound, except with probability exponentially small in $k+\card{S}\log\card{\Sigma}$ we have,
$$\card{\frac{1}{\card{X'}}\sum_{x\in X'} g_{S,h,x} - bias_{S,h}}\leq \gamma.$$
The lemma follows from a union bound over all $h:S\to\Sigma$.

\remove{
By applying the Chernoff bound to each $x\in X'$ and $\sigma\in\Sigma$, for each $x\in X'$, except with probability $\delta$,
$$\card{\tilde{g}_{S,h,x} - g_{S,h,x}}\leq\gamma.$$ This happens independently for each $x\in X'$.
For every $X''\subseteq X'$, $\card{X''}= \gamma \card{X'}$, the probability that for every $x\in X''$ we have $\card{\tilde{g}_{S,h,x} - g_{S,h,x}}>\gamma$ is at most $\delta^{\gamma \card{X'}}$. By taking a union bound over all $\binom{\card{X'}}{\gamma\card{X'}}\leq (e/\gamma)^{\gamma\card{X'}}$ possible $X''$, the probability that for at least $\gamma\card{X'}$ of the $x\in X'$ we have $\card{\tilde{g}_{S,h,x} - g_{S,h,x}}>\gamma$ is at most $(e\delta/\gamma)^{\gamma \card{X'}}$. Otherwise, we have $$\card{\tilde{bias}_{S,h}^{X'} - bias_{S,h}^{X'}}\leq 2\gamma.$$
The lemma follows from a union bound over all $h:S\to\Sigma$.
}
\end{proof}

Using an algorithm for the biased coin problem we get an algorithm for finding a good labeling to a free game. The algorithm is described in Figure~\ref{alg:amplified-free}. 

\begin{algbox}{6.5in}{amplified-free}{An algorithm for finding a good labeling to a free game $\mathcal{G}$ with $val(\mathcal{G})\geq 1-\varepsilon_0$. The error probability of the algorithm is exponentially small in $\card{X}\card{\Sigma}$.}
	\Procname{\textsc{Find-Labeling}$(\mathcal{G}=(G=(X,Y,X\times Y),\Sigma,\set{\pi_e}),\varepsilon_0,\varepsilon)$}
	\li Set $s= \ceil{\log(\card{\Sigma}/\varepsilon^2)/\varepsilon^2}$.
	\li Set $i_f= \log((\card{X}\card{\Sigma}+s\log\card{\Sigma})/\varepsilon^2)+\Theta(\log\log((\card{X}\card{\Sigma}+s\log \card{\Sigma})/\varepsilon))$;\;\; $\beta = \varepsilon/i_f$.
	\li Set $i_0= \log(s\log\card{\Sigma}/\beta^2)+\Theta(1)$.
	\li Sample $S\subseteq X$, $\card{S} = s$.
	\li \For $i=i_0,i_0+1,\ldots,i_f$ \Do
		\li Set $k=2^i$.
		\li \For all $h:S\to\Sigma$ \Do
					\li \textsc{Free-Toss-Coin}$(\mathcal{G},S,h,\beta^2 k,\gamma = \beta/80)$.
		\End
		\li If the fraction of heads is less than $1 -\varepsilon_0 - 2\varepsilon - i\beta$ for all $h$, restart.
	\End
	\li\label{free:check}\Return labeling $f_{S,h,X}$, $f_{S,h,Y}$ with value at least $1-\varepsilon_0 - 3\varepsilon$, if such exists.
\end{algbox}

The algorithm proves part of Theorem~\ref{t:free} repeated here for convenience. 
\begin{thm}\label{l:amplified-free}
Given a free game $\mathcal{G}$ with vertex sets $X,Y$, alphabet $\Sigma$, and $val(\mathcal{G})\geq 1-\varepsilon_0$ and given $\varepsilon > 0$, the algorithm \textsc{Find-Labeling} finds a labeling for $\mathcal{G}$ that achieves value at least $1-\varepsilon_0-3\varepsilon$ except with probability exponentially small in $\card{X}\card{\Sigma}$. The algorithm runs in time $\tilde{O}(\card{X}\card{Y}\card{\Sigma}^{O((1/\varepsilon^2)\log(\card{\Sigma}/\varepsilon))})$.
\end{thm}
In the remainder of the section we construct an oblivious verifier for free games and prove the rest of Theorem~\ref{t:free}, namely show a deterministic non-uniform algorithm.

\subsection{A Sketch for Free Games}

In this section we show that free games can be sketched using $\tilde{O}(\card{X}\card{\Sigma}^2\poly(1/\varepsilon))$ bits (as opposed to $O(\card{X}\card{Y}\card{\Sigma}^2)$ bits needed to describe the entire input game). The idea is to store the sub-game $\mathcal{G}_R$ induced on a small and carefully chosen set $R\subseteq Y$, that is, store all the constraints of the form $\pi_{(x,y)}$ for $x\in X$ and $y\in R$. We will show that this allows us to estimate the bias of every coin, as well as the value of the labeling induced by the coin on the random samples the algorithm makes.

Let $x\in X$. Let $g_{S,h,x}^R$ be the maximum over $\sigma\in\Sigma$ of the fraction of vertices $y\in R$ such that $(\sigma,f_{S,h,Y}(y))\in\pi_{(x,y)}$. We use the notation $bias^{X,R}_{S,h}\defeq (1/\card{X})\sum_{x\in X} g^R_{S,h,x}$ and $bias^{X',R}_{S,h}\defeq (1/\card{X'})\sum_{x\in X'} g^R_{S,h,x}$ for the bias of $S$, $h$ as witnessed by $R$. %Let $g_{S,h,x}$ be the maximum over $\sigma\in\Sigma$ of the fraction of %vertices $y\in Y$ such that $(\sigma,f_{S,h,Y}(y))\in\pi_{(x,y)}$.

\begin{lemma}[Free game sketch]\label{l:free-sketch} Let $s$ be as in Figure~\ref{alg:amplified-free}. Then, there exists $R\subseteq Y$, $\card{R} =\ceil{ s(s+1)\log\card{\Sigma}\log\card{X}/\gamma^2}$, such that for all $S\subseteq X$, $\card{S} = s$, for all $h:S\to\Sigma$, for every $x\in X$ we have $$\card{g_{S,h,x}^R - g_{S,h,x}} \leq \gamma.$$ 
As a result, $\card{bias_{S,h}^{X,R} - bias_{S,h}}\leq \gamma$, and, for all $X'\subseteq X$,
$\card{bias_{S,h}^{X',R} - bias_{S,h}^{X'}}\leq \gamma$.
\end{lemma}
\begin{proof}
Pick uniformly at random $R\subseteq Y$ of the specified size. Let $S\subseteq X$, $\card{S}=s$, $h:S\to\Sigma$, $x\in X$. By applying a Chernoff bound for every $\sigma\in \Sigma$ and taking a union bound over $\sigma\in\Sigma$, we get that $\card{g_{S,h,x}^R - g_{S,h,x}} \leq \gamma$ except with probability smaller than $\card{\Sigma}^{-s}\card{X}^{-(s+1)}$. By a union bound over all $S$, $h$ and $x$, we get that there exists $R\subseteq Y$ of the specified size such that $\card{g_{S,h,x}^R - g_{S,h,x}} \leq \gamma$ always holds. The claims about bias follow.
\end{proof}

\subsection{Oblivious Verifier for Free Games}

We design an oblivious verifier for \textsc{Find Labeling}, which we call \textsc{Oblivious-Verifier-Free-Execution}. The verifier gets the sketch of the input and the randomness of the algorithm \textsc{Find Labeling}, and follows the execution of the algorithm by guessing when it decides to restart. The final verifier, \textsc{Oblivious-Verifier-Free-Game}, checks all possible guesses. During its run \textsc{Oblivious-Verifier-Free-Execution} maintains $counter$ recording the low probability events it witnessed. If $counter$ ever reaches a value larger than $\card{X}\card{\Sigma}$, the verifier rejects. The verifier also maintains $\mathcal{H}$ the family of all free games consistent with the execution so far. If $\mathcal{H}$ becomes empty, the execution is designated as {\em infeasible}. Initially, $counter = 0$ and $\mathcal{H}$ contains all free games consistent with the sketch.
The final verifier checks that, no matter what were the guesses, all feasible executions of \textsc{Oblivious-Verifier-Free-Execution} accept. 
We argue that since the algorithm has error probability that is exponentially small in $\card{X}\card{\Sigma}$, the probability that the final verifier rejects is exponentially small in $\card{X}\card{\Sigma}$ as well. 

The verifier for a single execution (a single set of guesses) is described in Figure~\ref{alg:verifier-free}. Note that it uses the parameters $i_0,i_f,\beta$ of the algorithm. The final verifier is described in Figure~\ref{alg:verifier-final-free}.

%Each time a restart happens, the observed fraction of heads deviates %significantly from the true bias in the restart phase or in the previous phase. %A deviation can occur either because $X'$ does not reflect the true bias (i.e., %$bias_{S,h}^{X'}$ deviates significantly from $bias_{S,h}$), or because the %estimate of the bias that $X'$ reflects is incorrect (i.e., %$\tilde{bias}_{S,h}^{X'}$ deviates significantly from $bias_{S,h}^{X'}$). The %first case can be detected using the sketch, whereas in the second case we %expect the family $\mathcal{H}$ to shrink considerably since it is an event %with very low probability (here we use the choice of $\delta$).

\begin{algbox}{6.5in}{verifier-free}{An oblivious verifier for a single execution of \textsc{Find-Labeling} (an execution is defined by the outcomes of guesses). The final oblivious verifier checks that all feasible executions are accepted.}
	
\Procname{\textsc{Oblivious-Verifier-Free-Execution}$(\mathcal{H},\mathcal{G}_R,
\varepsilon_0,\varepsilon,r,counter)$} 
	\li\label{s:free-init}\If $\mathcal{H} = \emptyset$ \Then
		\li \Return infeasible.
	\End
	\li \If $counter > \card{X}\card{\Sigma}$ \Then
			\li \Return reject.
	\End
	\li Extract from $r$ the sample $S\subseteq X$ of the algorithm.
	\li \If $\max_{h} bias^{X,R}_{S,h}< 1 - \varepsilon_0 - 3\varepsilon - \beta/2 + \gamma$ \Then
			\li\label{s:free:+1}$counter \leftarrow counter + \log(1/\varepsilon)$
	\End
	\li\For $i=i_0,i_0+1,\ldots,i_f$ \Do
		\li $k\leftarrow 2^i$.
		\li \For all $h:S\to\Sigma$ \Do
			\li Extract from $r$ the sample $X'\subseteq X$ of the algorithm.
			\li Compute $bias^{X',R}_{S,h}$.
		\End
		\li \If $\exists h, \; \card{bias^{X',R}_{S,h} - bias^{X,R}_{S,h}} > 3\gamma$ \Then
					\li\label{s:free:+k}$counter\leftarrow counter+\beta^2 k - s\log \card{\Sigma}$.
				\End
		\li Guess if $\max_{h}bias_{S,h}^{X'}< 1-\varepsilon_0 - 2\varepsilon - i\beta$ and update $\mathcal{H}$ accordingly.
		\li\label{s:free-guess}\If guessed true  \Then
				\li Restart maintaining $\mathcal{H}$ and $counter$.
		\End
	\End
	\li\label{s:output}\Return accept iff $\max_{h} bias^{X,R}_{S,h}\geq 1 - \varepsilon_0 - 3\varepsilon - \gamma$.
\end{algbox}

\begin{algbox}{6.5in}{verifier-final-free}{The final oblivious verifier for free games.}
	\Procname{\textsc{Oblivious-Verifier-Free-Game}$(\mathcal{G}_R,\varepsilon_0,\varepsilon,randomness)$} 
	\li Let $\mathcal{H}$ contain all the free games that are consistent with $\mathcal{G}_R$.
	\li Try all guesses in \textsc{Oblivious-Verifier-Free-Execution}$(\mathcal{H},\mathcal{G}_R,\varepsilon_0,\varepsilon,randomness,0)$.
	\li Accept iff all feasible executions accept.
\end{algbox}

Next we analyze the oblivious verifier. We start by showing that when it accepts, \textsc{Find-Labeling} finds a high quality labeling (even if slightly of lower quality than in Figure~\ref{alg:amplified-free}).
\begin{lemma}
If \textsc{Oblivious-Verifier-Free-Game} accepts on the sketch of a game $\mathcal{G}$ and on randomness $r$, then \textsc{Find-Labeling}\begin{footnote}{Note that we argue about a version of \textsc{Find-Labeling} that checks a relaxed condition that is satisfied by a coin of bias $1- \varepsilon_0 - 3\varepsilon-2\gamma$ (as opposed to the version of Figure~\ref{alg:amplified-free}).}\end{footnote} produces a labeling that satisfies at least $1-\varepsilon_0 - 3\varepsilon - 2\gamma$ fraction of the edges when invoked on $\mathcal{G}$ and $r$ and with the same parameters $\varepsilon_0,\varepsilon$.
\end{lemma}
\begin{proof}
Let $\mathcal{G}$ be the input game to \textsc{Find-Labeling}. If \textsc{Oblivious-Verifier-Free-Game} accepts with the sketch $\mathcal{G}_R$, then, in particular, the execution of \textsc{Oblivious-Verifier-Free-Execution} with the guesses that correspond to the run of \textsc{Find-Labeling} accepts. By Lemma~\ref{l:free-sketch}, \textsc{Find-Labeling} finds a labeling that satisfies $1-\varepsilon_0 - 3\varepsilon - 2\gamma$ fraction of the edges. 
\end{proof}

In order to argue that \textsc{Oblivious-Verifier-Free-Game} rejects with probability exponentially small in $\card{X}\card{\Sigma}$ we prove the following.
\begin{lemma}\label{l:free-soundness} The following hold:
\begin{itemize}
\item For any guesses, the probability that \textsc{Oblivious-Verifier-Free-Execution} rejects is exponentially small in $\card{X}\card{\Sigma}$.

\item \textsc{Oblivious-Verifier-Free-Execution} makes a number of guesses that is sufficiently smaller than $\card{X}\card{\Sigma}$ (the ration between the number of guesses and $\card{X}\card{\Sigma}$ can be made arbitrarily small by multiplying $\varepsilon$ by a suitable constant and by increasing the constant term of $i_0$).
\end{itemize}
\end{lemma}
The correctness of the final verifier follows from a union bound over all possible guesses. If the number of guesses is sufficiently small, then there is an exponentially small probability in $\card{X}\card{\Sigma}$ that the final verifier rejects.

In order to prove Lemma~\ref{l:free-soundness}, we show the following invariants.
\begin{lemma}\label{l:free-invariants}
The following invariants are maintained throughout the run of \textsc{Oblivious-Verifier-Free-Execution}:
\begin{enumerate}
\item\label{i:free-1} $\mathcal{H}$ consists of all the free games that are consistent with the sketch and the guesses so far.
\item\label{i:free-2} The probability that $counter=c$ is exponentially small in $c$.
\item\label{i:free-3} If \textsc{Oblivious-Verifier-Free-Execution} restarts in phase $i$, then either in the current phase or in the previous $counter$ increases by at least $\beta^2 2^{i-1} - s\log\card{\Sigma}$, or $i=i_0$ and $counter$ increases by $\log(1/\varepsilon)$.
\end{enumerate} 
\end{lemma}
\begin{proof}
It's clear that Invariant~\ref{i:free-1} holds.
Next we show that Invariant~\ref{i:free-2} holds. Initially $counter$ is set to $0$. The increase in Step~\ref{s:free:+1} is justified by Lemma~\ref{l:Russell} and the design
of the sketch (Lemma~\ref{l:free-sketch}). 
The increase in Step~\ref{s:free:+k} is justified by Lemma~\ref{l:free-sample} and the design
of the sketch (Lemma~\ref{l:free-sketch}). 

Next we show that Invariant~\ref{i:free-3} holds.
Suppose that there is a restart in phase $i$. Let $k=2^i$.
Lemma~\ref{l:simulated-coin} ensures that either in the current phase or in the previous one $\max_h bias_{S,h}^{X'}$ deviates from $\max_h bias_{S,h}$ by more than an additive $\beta/2$, or it's the first phase and $\max_h bias_{S,h}$ is smaller than $1-\varepsilon_0-2\varepsilon - \beta/2$. We handle both cases:

\begin{enumerate}
\item 
Suppose that this is the first phase and $\max_h bias_{S,h}<1-\varepsilon_0 - 2\varepsilon - \beta/2$. By the design of the sketch (Lemma~\ref{l:free-sketch}), we have $\max_h bias_{S,h}^{X,R}<1-\varepsilon_0 - 2\varepsilon -\beta/2 +\gamma$, and hence $counter$ increases as required in Step~\ref{s:free:+1}.

\item Suppose that either in the current phase or in the previous one $\max_h bias_{S,h}^{X'}$ deviates from $\max_h bias_{S,h}$ by more than an additive $\beta/2$. By the design of the sketch (Lemma~\ref{l:free-sketch}), either in this phase or in the previous there exists $h:S\to\Sigma$ such that  $\card{bias_{S,h}^{X',R}-bias_{S,h}^{X,R}}>3\gamma$. In this case, $counter$ increases appropriately in Step~\ref{s:+k} of the appropriate phase. 
\end{enumerate}
\end{proof}

We can now prove Lemma~\ref{l:free-soundness} from the invariants of Lemma~\ref{l:free-invariants}.

\begin{proof}(of Lemma~\ref{l:free-soundness} from Lemma~\ref{l:free-invariants})
Let us show that the first item in Lemma~\ref{l:free-soundness} follows from Invariant~\ref{i:free-2}. The verifier only rejects if $counter > \card{X}\card{\Sigma}$, or if it reached Step~\ref{s:output} and rejected. 
The invariant ensures that the probability that $counter > \card{X}\card{\Sigma}$ is exponentially small in $\card{X}\card{\Sigma}$. 
Theorem~\ref{l:amplified-free} ensures that the probability that $\max_h bias_{S,h}< 1 -\varepsilon_0 - 3\varepsilon$ is exponentially small in $\card{X}\card{\Sigma}$. The design of the sketch (Lemma~\ref{l:free-sketch}) ensures that if $\max_h bias_{S,h}\geq 1 -\varepsilon_0 - 3\varepsilon$, then  $\max_h bias_{S,h}^{X,R}\geq  1 -\varepsilon_0 - 3\varepsilon -\gamma$.

The second item follows from Invariant~\ref{i:free-3}, since $counter < \card{X}\card{\Sigma}$, and, $counter$ increases in large increments: Invariant~\ref{i:clique-3} ensures that the increments only depend on $\varepsilon$ and $i_0$, and can be made arbitrarily large by multiplying $\varepsilon$ by a small constant and by adding to $i_0$ a large constant.
\end{proof}

A non-uniform deterministic algorithm for free games follows, concluding the proof of Theorem~\ref{t:free}.
\begin{thm}\label{t:free-det}
There is a deterministic non-uniform algorithm that given a free game $\mathcal{G}$ with vertex sets $X,Y$, alphabet $\Sigma$ and $val(\mathcal{G})\geq 1-\varepsilon_0$, finds a labeling to the vertices that satisfies $1-\varepsilon_0 - O(\varepsilon)$ fraction of the edges. The algorithm runs in time $\tilde{O}(\card{X}\card{Y} \card{\Sigma}^{O((1/\varepsilon^2)\log(\card{\Sigma}/\varepsilon))})$.
\end{thm}
We remark that there is an alternative way to prove Theorem~\ref{t:free-det} similarly to Remark~\ref{r:clique-alternative}.

\section{From List Decoding to Unique Decoding of Reed-Muller Code}

% complete: sampling, Johnson

\subsection{A Randomized Algorithm With Error Probability Roughly $1/\card{\field}$}\label{s:RM-corr}

Let $\field$ be a finite field, and let $\dims > 3$ and $\degree<\card{\field}$ be natural numbers.
First we describe a randomized algorithm with error probability $\card{\field}^{-\Omega(1)}$ for reducing the Reed-Muller list decoding problem with parameters $\field$, $\dims$, $\degree$ to the Reed-Muller unique decoding problem. The algorithm is based on the idea of self-correction (see, e.g.,~\cite{STV} and the references there). The algorithm is described in Figure~\ref{alg:self-corr}. On input $f:\vecspace\to\field$ it picks a random line $\ell$ and finds all the polynomials that agree with $f$ on about $\rho$ fraction of the points in $\ell$. We'll show that if $f$ agrees with an $\dims$-variate polynomial $p$ on $\rho$ fraction of the points in $\vecspace$, then, except with probability roughly $1/\card{\field}$ over the choice of $\ell$, there is about $\rho$ fraction of the points on $\ell$ on which $f$ agrees with $p$. Hence, the restriction of $p$ to $\ell$ is likely to be one of the polynomials that the algorithm finds.
The algorithm outputs a list of functions $g_1,\ldots,g_k:\vecspace\to\field$. Each $g_i$ corresponds to one of the polynomials in the line list. The algorithm computes each $g_i$ by iterating over all $z\in\vecspace$ and considering the plane $s$ spanned by $\ell$ and $z$. 
Again, except for fraction roughly $1/\card{\field}$ of the $z\in\vecspace$, there is about $\rho$ fraction of the points on $s$ on which $f$ agrees with $p$. The algorithm sets $g_i(z)=f(z)$ if there is a unique polynomial that agrees with $f$ on about $\rho$ fraction of the points in $s$ and with $g_i$'s polynomial on $\ell$.

%there are at most $2/\rho$ polynomials $p_1,\ldots,p_l$ of degree at most $d$ %that are $\rho$-close to $f$. We can apply {\em self-correction} (aka {\em %random self-reducibility}) to evaluate  

\begin{algbox}{6.5in}{self-corr}{A randomized algorithm with error probability $\card{\field}^{-\Omega(1)}$ that finds $g_1,\ldots,g_k:\vecspace\to\field$, $k\leq O(1/\rho)$, such that for every polynomial $p$ of degree at most $\degree$ that agrees with $f$ on $\rho$ fraction of the points in $\vecspace$ there is $g_i$ that agrees with $p$ on at least $1-\epsilon$ fraction of the points.}
	\Procname{\textsc{Self-Correct}$(f,\rho,\epsilon)$}
	\li\label{s:sample}Pick uniformly at random $x,y\in\vecspace$, $y\neq\vec{0}$.
	\li\label{s:uni-list}Find all univariate $p^{(1)}_{x,y},\ldots,p^{(k)}_{x,y}$ so
$\card{\sett{t\in\field}{f(x+ty) = p^{(j)}_{x,y}(t)}}\geq(\rho-\epsilon)\cdot\card{\field}$.
\li \For $z\in\vecspace$ such that $z-x,y$ are independent \Do 
		%\li Let $S$ be the subspace $\sett{z+tx+sy}{t,s\in\field}$.
		\li\label{s:bi-list}Find $q^{(1)},\ldots,q^{(k')}$: $\card{\sett{t_1,t_2\in\field}{f(x+t_1 y+t_2 (z-x)) = q^{(j)}(t_1,t_2)}}\geq(\rho-\epsilon)\cdot\card{\field}^2$.
		\li \For $1\leq i\leq k$ \Do 
			\li \If $\exists !\; 1\leq j\leq k',\; p^{(i)}_{x,y}(t)= q^{(j)}(t,0)$ for all $t\in\field$ \Then
			\li Set $g_{i}(z) = q^{(j)}(0,1)$.
			\End
		\End
	\End
	\li \Return $g_1,\ldots,g_k$.
\end{algbox}

Steps~\ref{s:uni-list} and~\ref{s:bi-list} that require list decoding of Reed-Solomon code can be performed in time $\poly(\card{\field})$. Therefore, the run time of the algorithm is $O(\card{\vecspace}\poly(\card{\field}))$. A standard choice of parameters is $\card{\field} = \poly\log\card{\vecspace}$, and it leads to a run-time of $\tilde{O}(\card{\vecspace})$.
Next we prove the correctness of the algorithm.

We'll need the following lemma about list decoding for polynomials.
\begin{lemma}[List decoding]\label{l:short-list}
Fix a finite field $\field$ and natural numbers $\dims$ 
and $\degree<\card{\field}$. Let
$f:\vecspace\rightarrow\field$. Then, for any $\rho \geq
2\sqrt{\frac{\degree}{\card{\field}}}$, if
$q_1,\ldots,q_k:\vecspace\rightarrow\field$ are {\sf different}
polynomials of degree at most $\degree$, and for every $1\leq
i\leq k$, the polynomial $q_i$ agrees with $f$ on at least
$\rho$ fraction of the points, i.e.,
$\Prob{x\in\vecspace}{q_i(x) = f(x)}\geq
\rho$, then $k\leq\frac{2}{\rho}$.
\end{lemma}
\begin{proof}
Let $\rho \geq 2\sqrt{\frac{\degree}{\card{\field}}}$, and
assume on way of contradiction that there exist $k =
\floor{\frac{2}{\rho}}+1$ different polynomials
$q_1,\ldots,q_k:\vecspace\to\field$ as stated.

For every $1\leq i\leq k$, let
$A_i\defeq\sett{x\in\vecspace}{q_i(x)=f(x)}$. By
inclusion-exclusion,
$$\card{\vecspace}\geq\card{\bigcup_{i=1}^{k} A_i} \geq
\sum_{i=1}^{k}\card{A_i} - \sum_{i\neq j}\card{A_i\cap A_j}$$ By
Schwartz-Zippel, for every $1\leq i\neq j\leq k$, $\card{A_i\cap
A_j}\leq \frac{\degree}{\card{\field}}\cdot\card{\vecspace}$.
Therefore, by the premise,
$$\card{\vecspace}\geq
k\rho\card{\vecspace} -
\binom{k}{2}\frac{\degree}{\card{\field}}\card{\vecspace}$$ On one
hand, since $k>\frac{2}{\rho}$, we get $k\rho> 2$. On the
other hand, since
$\frac{2}{\rho}\leq\sqrt{\frac{\card{\field}}{\degree}}$ and
$\degree\leq\card{\field}$, we get
$\binom{k}{2}\leq\frac{\card{\field}}{\degree}$. This results in a
contradiction.
\end{proof}

Let $p$ be an $\dims$-variate polynomial of degree at most $\degree$ over $\field$ that agrees with $f$ on at least $\rho$ fraction of the points $x\in\vecspace$. We will show that most likely one of the $g_i$'s that the algorithm outputs is very close to $p$.

In the following lemma we argue that the restriction of $p$ to the line defined by $x$ and $y$ is likely to appear in the line list.
\begin{lemma}[Sampling]\label{l:line-sample} 
%Assume that $\epsilon \geq 4\degree/(\rho\card{\field})$.
Except with probability $\rho/(\epsilon^2\card{\field})$ over the choice of $x$ and $y$, for at least $(\rho-\epsilon)\card{\field}$ elements $t\in\field$ we have $f(x+ty)=p(x+ty)$.
\end{lemma}
\begin{proof}
The lemma follows from a second moment argument. Let $A\subseteq\vecspace$, $\card{A} = \rho\card{\vecspace}$ contain elements $z\in\vecspace$ such that $f(z)=p(z)$. For uniform $x,y\in\vecspace$, for $t\in\field$ let $X_t$ be an indicator random variable for $x+ty\in A$. Let $X=(1/\card{\field})\sum X_t$. We have $\Expc{}{X} = \rho$. For every $t\neq t'\in\field$ we have that $X_t$ and $X_{t'}$ are independent. Hence, $\Expc{}{X^2}= (1/\card{\field})^2\sum_{t, t'}\Expc{}{X_t X_{t'}} = (1/\card{\field})^2\sum_{t}\Expc{}{X_t} = \rho/\card{\field}$. By Chebychev inequality,
$$\Prob{}{\card{X - \rho}>\epsilon}\leq \frac{\Expc{}{X^2}}{\epsilon^2} = \frac{\rho}{\epsilon^2\card{\field}}.$$
The lemma follows.
\end{proof}

The same holds for the planes defined by most $z\in\vecspace$.
\begin{lemma}[Sampling]\label{l:subspace-sample} 
%Assume that $\epsilon \geq 4\degree/(\rho\card{\field})$.
Except with probability $\rho/(\epsilon^2\card{\field})$ over the choice of $z$, $x$ and $y$, for at least $(\rho-\epsilon)\card{\field}^2$ elements $t_1,t_2\in\field$, we have $f(x+t_1y+t_2(z-x))=p(x+t_1y+t_2(z-x))$.
\end{lemma}

From Lemmas~\ref{l:line-sample} and~\ref{l:subspace-sample} it follows that except with probability roughly $1/\card{\field}$ restrictions of $p$ appear both in the line list and in most planes lists. Next we'll argue that for most $z\in\vecspace$ it's unlikey that the restriction of $p$ to the plane is not the unique polynomial in the plane list that agrees with $p$ on the line.

\begin{lemma}\label{l:RM-uniqueness}
The probability over the choice of $x$, $y$ and $z$ that there are $1\leq j< i\leq k'$ such that $q^{(j)}(t,0)\equiv q^{(i)}(t,0)$ even though $q^{(j)}\not\equiv q^{(i)}$, is at most $4\degree/((\rho-\epsilon)^2\card{\field})$.
\end{lemma}
\begin{proof}
Fix $1\leq j< i\leq k'$.
Suppose that one first picks the three dimensional subspace $s$ and then picks $x,y,z\in\vecspace$ such that $\sett{x+t_1y+t_2(z-x)}{t_1,t_2\in\field}$. The probability over the choice of $x$ and $y$, that $q^{(j)}$ agrees with $q^{(i)}$ on the line $\sett{x+ty}{t\in\field}$ is at most $\degree/\card{\field}$. By Lemma~\ref{l:short-list}, there are at most $2/(\rho-\epsilon)$ polynomials in the list of $s$. Taking a union bound over all choices of $1\leq j< i\leq k'$ results in the lemma.
\end{proof}

Hence, except with probability $O(\delta)$ over $x$, $y$ for $\delta = \max\set{\degree/(\rho^2\card{\field}),\rho/(\epsilon^2\card{\field})}$, the restriction of $p$ to the line $\sett{x+ty}{t\in\field}$ appears as $p_{x,y}^{(i)}$, and, moreover, $g_i$ agrees with $p$ on all but $O(\delta)$ fraction of the $z$ (note that only a small fraction $\card{\field}^{2}/\card{\vecspace}$ of the $z\in\vecspace$ satisfy that $z-x,y$ are dependent).

\subsection{Finding Approximate Codewords as Finding a Biased Coin}

We describe an analogy between finding a list of approximate polynomials and finding a biased coin. We think of picking a line and finding a list decoding of $f$ on the line as picking a coin. The coin picking algorithm is described in Figure~\ref{alg:pick-poly-coin}. We think of sampling $z\in\vecspace$ and checking whether the line list decoding is consistent with the list decoding on the subspace defined by $z$ and the line as a coin toss that falls on ``heads'' if there is consistency. The coin tossing algorithm is described in Figure~\ref{alg:toss-poly-coin}. 

\begin{algbox}{6.5in}{pick-poly-coin}{A coin corresponds to a line in $\vecspace$ and the list decoding of $f$ on the line.}
	\Procname{\textsc{RM-Pick-Coin}$(f,\rho,\epsilon)$}
	\li Pick uniformly at random $x,y\in\vecspace$, $y\neq\vec{0}$.
	\li Find univariate polynomials $p^{(1)}_{x,y},\ldots,p^{(k)}_{x,y}$ so
$\card{\sett{t\in\field}{f(x+ty) = p^{(j)}_{x,y}(t)}}\geq(\rho-\epsilon)\cdot\card{\field}$.	
	\li \Return $x,y,p^{(1)}_{x,y},\ldots,p^{(k)}_{x,y}$.
\end{algbox}

\begin{algbox}{6.5in}{toss-poly-coin}{A coin toss corresponds to picking a uniform $z\in\vecspace$ and checking whether the line list decoding is consistent with the list decoding on the subspace defined by the line and $z$.}
	\Procname{\textsc{RM-Toss-Coin}$(f,\rho,\epsilon,x,y,p^{(1)}_{x,y},\ldots,p^{(k)}_{x,y})$}
	\li Pick uniformly $z\in\vecspace$ independent of $x,y$.
	\li Find $q^{(1)},\ldots,q^{(k')}$ so $\card{\sett{t_1,t_2\in\field}{f(x+t_1y+t_2 (z-x)) = q(t_1,t_2)}}\geq(\rho-\epsilon)\cdot\card{\field}^2$.
	\li \For $1\leq i\leq k$ \Do 
			\li \If $\neg\exists ! 1\leq j\leq k',\; p^{(i)}_{x,y}(t)\equiv q^{(j)}(t,0)$ \Then
			\li \Return ``tails''.
			\End
		\End
	\li \Return ``heads''.
\end{algbox}

\begin{algbox}{6.5in}{recover}{An algorithm that uses a biased coin (given by $x,y,p^{(1)}_{x,y},\ldots,p^{(k)}_{x,y}$) to find a short list of functions $g_1,\ldots,g_k:\vecspace\to\field$ such that for every polynomial $p$ of degree at most $\degree$ that agrees with $f$ on $\rho$ fraction of the points in $\vecspace$ there is $g_i$ that agrees with $p$ on at least $1-\epsilon$ fraction of the points.}
\Procname{\textsc{RM-Interpolate}$(f,\rho,\epsilon,x,y,p^{(1)}_{x,y},\ldots,p^{(k)}_{x,y})$}
\li \For $z\in\vecspace$ independent of $x,y$ \Do 
		%\li Let $S$ be the subspace $\sett{z+tx+sy}{t,s\in\field}$.
		\li Find $q^{(1)},\ldots,q^{(k')}$ so $\card{\sett{t_1,t_2\in\field}{f( x+t_1 y+t_2 (z-x)) = q(t_1,t_2)}}\geq(\rho-\epsilon)\cdot\card{\field}^2$.
		\li \For $1\leq i\leq k$ \Do 
			\li \If $\exists !\; 1\leq j\leq k',\; p^{(i)}_{x,y}(t)= q^{(j)}(t,0)$ for all $t\in\field$ \Then
			\li Set $g_{i}(z) = q^{(j)}(0,1)$.
			\End
		\End
	\End
	\li \Return $g_1,\ldots,g_k$.
\end{algbox}

Lemmas~\ref{l:line-sample},~\ref{l:subspace-sample} and~\ref{l:RM-uniqueness} ensure that a biased coin is picked with at least a constant probability for sufficiently large $\rho>\epsilon>0$ and sufficiently small $\degree<\card{\field}$.
Note that both picking a coin and tossing it take short time $\poly(\card{\field})$. Hence, if we use $\tilde{O}(\card{\vecspace})$ coin tosses to find a biased coin and $\card{\field} = \poly\log\card{\vecspace}$, then we get an algorithm with $\tilde{O}(\card{\vecspace})$ run-time. 
Moreover, using a biased coin one can compute a short list of approximate polynomials as in Figure~\ref{alg:recover}.

\begin{lemma}
Assume that $x,y,p^{(1)}_{x,y},\ldots,p^{(k)}_{x,y}$ define a biased coin (namely, a coin that falls on ``heads'' with probability at least $1-O(\delta)$ for $\delta$ as in Section~\ref{s:RM-corr}), and that $\epsilon$ is smaller than $O(\delta)$ from Section~\ref{s:RM-corr}. Then, for every $\dims$-variate polynomial $p$ of degree at most $\degree$ over $\field$ there exists $g_i$ in the list computed by \textsc{RM-Interpolate} that agrees with $p$ on at least $1-O(\delta)$ fraction of the points.
\end{lemma}
\begin{proof} Let $p$ be an $\dims$-variate polynomial of degree at most $\degree$ over $\field$ that agress with $f$ on at least $\rho$ fraction of the points $x\in\vecspace$. 
Assume on way of contradiction that none of $p^{(1)}_{x,y},\ldots,p^{(k)}_{x,y}$ is $p$ restricted to the line $\sett{x+ty}{t\in\field}$ (otherwise, we are done as we argued in Section~\ref{s:RM-corr}). Take $\epsilon$ sufficiently smaller than $O(\delta)$ in the lemma.
For at least $\epsilon$ fraction of the $z\in\vecspace$ such that $z-x,y$ are independent, the fraction of points $x+t_1y+t_2(z-x)$ with $t_1,t_2\in\field$ on which $f$ agrees with $p$, is at least $\rho-\epsilon$. For those choices of $z$, the coin falls on ``tails'', hence the bias of the coin is at most $1-\epsilon$, which is a contradiction.
\end{proof}

Therefore, \textsc{Find-Biased-Coin} when using \textsc{RM-Pick-Coin} and \textsc{RM-Toss-Coin}, and when followed by \textsc{RM-Interpolate} to obtain the list of approximate polynomials from the coin, solves the list decoding to unique decoding problem for the Reed-Muller code in time $\tilde{O}(\card{\vecspace}\poly(\card{\field}))$ and with error probability exponentially small in $\card{\vecspace}\log\card{\field}$.
Since the input $f,\rho,\epsilon$ is of size $\card{\vecspace}\log\card{\field}$, we also get a deterministic non-uniform algorithm that runs in similar time.
For convenience, we repeat Theorem~\ref{t:RM} that we just proved.
\begin{theorem}
Let $\field$ be a finite field, let $\degree$ and $\dims>3$ be natural numbers and let $0<\rho,\epsilon<1$, such that $\degree\leq \card{\field}/10$, $\epsilon >\sqrt[3]{2/\card{\field}}$ and $\rho>\epsilon + 2\sqrt{\degree/\card{\field}}$. 
There is a randomized algorithm that given $f:\vecspace\to\field$, finds a list of $l=O(1/\rho)$ functions $g_1,\ldots,g_l:\vecspace\to\field$, such that for every $\dims$-variate polynomial $p$ of degree at most $\degree$ over $\field$ that agrees with $f$ on at least $\rho$ fraction of the points $x\in\vecspace$, there exists $g_i$ that agrees with $p$ on at least $1-\epsilon$ fraction of the points $x\in\vecspace$. The algorithm has error probability exponentially small in $\card{\vecspace}\log\card{\field}$ and it runs in time $\tilde{O}(\card{\vecspace}\poly(\card{\field}))$. It implies a deterministic non-uniform algorithm with the same run-time. 
\end{theorem}

\section{Open Problems}

\begin{itemize}
\item We obtained efficient non-uniform deterministic algorithms. It would be very interesting to convert them to uniform algorithms.

\item What other algorithms can be derandomized using our method? Can more sophisticated sketching and sparsification techniques be used to handle algorithms on sparse graphs? The applications in this paper have Atlantic City algorithms that run in sub-linear time, but we do not think that the method is limited to such problems. It will be interesting to find concrete examples.
%where not only the algorithm that are derandomized run in at least linear time, %but also no sub-linear time algorithms

\item What lower bound can one prove on the number of coin tosses needed to find a biased coin? What if the target bias is not known, yet it is known that a large fraction of the coins achieve that target? Solving the latter would yield an algorithm for \textsc{Free Games} that handles games with general value, rather than value close to $1$. 

\item Can one derandomize algorithms using pseudorandom generators and make use of the fact that the pseudorandomness should look random only to a distinguisher that is a verifier for the algorithm? Can one use the existence of an oblivious verifier to construct better psuedorandom generators?

\item Are there deterministic algorithms for \textsc{Max-Cut} on dense graphs that run in time $\tilde{O}(\card{V}^2 + (1/\varepsilon)^{O(1/\gamma\varepsilon^2)})$ or even $O(\card{V}^2 + 2^{O(1/\gamma\varepsilon^2)})$ instead of $\tilde{O}(\card{V}^2(1/\varepsilon)^{O(1/\gamma\varepsilon^2)})$?
Recall that the randomized algorithm of Mathieu and Schudy~\cite{MS} runs in time $O(\card{V}^2 + 2^{O(1/\gamma^2\varepsilon^2)})$. Are there deterministic algorithms for (approximate) \textsc{Clique} that run in time $\tilde{O}(\card{V}^2 + 2^{O(1/(\rho^3\varepsilon^2))})$ instead of $\tilde{O}(\card{V}^2 2^{O(1/(\rho^3\varepsilon^2))})$?

\item The run-times of our algorithms have $\poly\log n$ factors coming from our algorithm for the biased coin problem and from the size of the sketches. Can they be eliminated?

\end{itemize}

\section*{Acknowledgements}
Dana Moshkovitz is grateful to Sarah Eisenstat for her collaboration during the long preliminary stages of this work. Many thanks to Scott Aaronson, Noga Alon, Bernard Chazelle, Shiri Chechik, Shayan Oveis Gharan, Oded Goldreich, David Karger, Guy Moshkovitz, Michal Moshkovitz, Richard Peng, Seth Pettie, Vijaya Ramachandran, Aaron Sidford, Dan Spielman, Bob Tarjan, Virginia Vassilevska Williams, Avi Wigderson and Ryan Williams for discussions.

\bibliographystyle{plain}
\bibliography{bibfile}

\end{document}